\documentclass[conference]{IEEEtran}
\IEEEoverridecommandlockouts

\usepackage{cite}
\usepackage{comment}
\usepackage{float}
\usepackage{amsmath,amssymb,amsfonts,amsthm, bbm}
\usepackage{algorithm}
\usepackage{algpseudocode}
\usepackage{graphicx}
\usepackage{subcaption}
\usepackage{textcomp}
\usepackage{xcolor}
\usepackage{tikz}
\usepackage{mathtools}
\usepackage[normalem]{ulem}
\usepackage[margin=0.75in]{geometry}
\usepackage{scalerel}
\usepackage{placeins}

\begin{document}

\newtheorem{hypotheses}{Hypotheses}
\newtheorem{problem}{Problem}
\newtheorem{example}{Example}
\newtheorem{definition}{Definition}
\newtheorem{assumption}{Assumption}
\newtheorem{theorem}{Theorem}
\newtheorem{lemma}{Lemma}
\newtheorem{corollary}{Corollary}[theorem]
\newtheorem{proposition}{Proposition}
\newtheorem*{remark}{Remark}
\newtheorem{conjecture}{Conjecture}
\numberwithin{assumption}{section}

\newcommand{\edit}[1]{\textcolor{black}{#1}}
\newcommand{\arh}[1]{{\color{blue}{[ARH: #1]}}}

\newcommand{\will}[1]{{\color{olive}{[Will: #1]}}}
\newcommand{\ignore}[1]{}
\newcommand{\diff}{\mathop{}\!\mathrm{d}}

\algnewcommand{\Inputs}[1]{%
  \State \textbf{Inputs:}
  \Statex \hspace*{\algorithmicindent}\parbox[t]{.8\linewidth}{\raggedright #1}
}
\algnewcommand{\Initialize}[1]{%
  \State \textbf{Initialize:}
  \Statex \hspace*{\algorithmicindent}\parbox[t]{.8\linewidth}{\raggedright #1}
}

\title{
Online Identification of Time-Varying Systems Using Excitation Sets and Change Point Detection 
    \thanks{*Chi Ho Leung and Philip E. Par\'e are with the Elmore Family School of Electrical and Computer Engineering, Purdue University, USA. Ashish R. Hota is with the Department of Electrical Engineering, IIT Kharagpur, India.  E-mail: leung61@purdue.edu, ahota@ee.iitkgp.ac.in, philpare@purdue.edu. This material is based upon work supported in part by the US-India Collaborative Research Program between the US National Science Foundation (NSF-ECCS \#2032258, \#2238388) and the Department of Science and Technology of India (via IDEAS TIH, ISI Kolkata).}
}

\author{Chi Ho Leung, Ashish R. Hota, and Philip E. Par\'{e}*}

\maketitle

\begin{abstract}
In this work, we first show that the problem of parameter identification is often ill-conditioned and lacks the persistence of excitation required for the convergence of online learning schemes. To tackle these challenges, we introduce the notion of optimal and greedy excitation sets which contain data points with sufficient richness to aid in the identification task. We then present the greedy excitation set-based recursive least squares algorithm to alleviate the problem of the lack of persistent excitation, and prove that the iterates generated by the proposed algorithm minimize an auxiliary weighted least squares cost function. When data points are generated from time-varying parameters, online estimators tend to underfit the true parameter trajectory, and their predictability deteriorates. To tackle this problem, we propose a memory resetting scheme leveraging change point detection techniques. Finally, we illustrate the performance of the proposed algorithms via several numerical case studies to learn the (time-varying) parameters of networked epidemic dynamics, and compare it with results obtained using conventional approaches.
\end{abstract}

\section{Introduction}

By continuously updating predictions with incoming data, practitioners can better manage uncertainty, optimize performance, and mitigate risks~\cite{hota2021closed}.
\edit{However, traditional adaptive identification methods, such as gradient descent and recursive least squares, often struggle due to inadequate excitation, poor identifiability, and difficulties tracking time-varying system parameters~\cite{islam2019recursive}. 
Persistent excitation, defined as the consistent information richness of input signals, is generally considered sufficient for adaptive estimation convergence, but is often excessively restrictive in practice~\cite{ljung1994global,ioannou2006adaptive,narendra1987persistent}. 
Similarly, while structural identifiability ensures that parameters can be theoretically estimated from noise-free data~\cite{hamelin2020observability,ljung1994global}, practical identifiability, the ability to reliably estimate parameters from noisy real-world data is crucial, yet challenging~\cite{wieland2021structural}.}

Contemporary approaches for overcoming the challenges of the lack of persistent excitation and practical identifiability include indirect adaptive control strategies~\cite{ioannou1988theory} and robust adaptive control methods~\cite{ioannou1996robust}. 
Efforts to mitigate the lack of uniform persistent excitation have also \edit{been} explored, relaxing this stringent requirement and introducing concepts like initial excitation~\cite{jha2019initial, dhar2022initial}. 
Additionally, techniques that incorporate second-order information into each parameter estimation step, such as recursive least squares (RLS) methods~\cite{cao2000directional, lai2024generalized, salgado1988modified}, have been developed. 
In this work, we propose an algorithm that combines the well-known RLS algorithm with the novel concept of an excitation set that is used to construct the main regressor. 
We explain the intuition behind this idea, and show that the estimates obtained by the proposed algorithm minimize an auxiliary cost function that weighs exciting data points with unity weight and disregards regressors that are not sufficiently exciting. 

\edit{Further}, we often encounter systems with time-varying parameters in practice. Previous research has addressed offline system identification problems with piece-wise constant time-varying parameters~\cite{fearnhead2006exact}. 
Online parameter estimation, however, remains challenging. 
Some solutions have been proposed through detecting non-stationary statistical properties in the prediction error of recursive least squares~\cite{mohseni2022recursive}. 
In this work, we combine the exponentially weighted average model~\cite{hunter1986exponentially} and the likelihood ratio test~\cite{hogg1977probability} to detect abrupt changes in system parameters through \edit{the} prediction error. 
Furthermore, we propose \edit{to combine} this change point detection and memory resetting algorithm \edit{with our learning algorithm as well as} a broad class of online estimation methods.

Lastly, recognizing the well-documented challenge of practical identifiability in epidemic models as highlighted in recent studies \cite{hamelin2020observability, prasse2022predicting}, we compare our results with existing approaches. 
We demonstrate the effectiveness of our algorithm for online parameter identification in time-varying epidemic models~\cite{kermack1927sir, giordano2020modelling} with noise. 
Specifically, we consider the susceptible-infected-susceptible~(SIS) and networked susceptible-infected-recovered~(SIR) models with time-varying parameters to illustrate the performance of our proposed algorithms.

This paper builds upon our previous work in~\cite{leung2023adaptive}, enhancing it with additional analysis, algorithms, and simulations to address the online estimation challenges present in a broad class of time-varying nonlinear systems beyond single-node SIS systems.
We organize the content as follows: in Section~\ref{chp:prelim}, we present the problem formulation and some necessary definitions for the rest of the paper.
In Section~\ref{sec:lack-pe-pi}, we give a sufficient condition for the lack of persistent excitation.
In Section~\ref{sec:exc-set}, we present the definition of an excitation set and greedy excitation set, which are used to develop the greedily-weighted recursive least squares algorithm.
In Section~\ref{sec:exc-set-rls}, we present the greedily-weighted recursive least squares algorithm as a way to overcome the lack of persistent excitation and practical identifiability assuming parameters are time-invariant.
In Section~\ref{sec:fault-dect}, we present a way to facilitate tracking of time-varying, piecewise-constant parameters through memory resetting upon detection of abrupt changes. 
In Section~\ref{sec:sim}, we demonstrate the performance of the proposed algorithms on \edit{SIS and networked SIR} models.

\subsection{Notations}
We denote a matrix $X$ to be positive and negative semi-definite by $X \succeq 0$ and $X \preceq 0$, respectively. 
The condition number of a matrix $X$, defined as the ratio of the largest singular value of \(X\) to the smallest singular value of \(X\), is denoted by \(\kappa(X)\). 
For some positive semi-definite matrix $A$, the weighted norm of a vector $x$ is defined as $\|x\|_{A} = \sqrt{x^\top Ax}$, and the \edit{2-norm is denoted by $\|\cdot\|$}. 
The spectral norm of a matrix $A \in \mathbb{R}^{m\times n}$ is denoted as $\|A\|_2$, which is also the maximum singular value of $A$, while $\|v\|_\infty\coloneqq \max(v)$ denote the infinity norm.
The maximum/minimum eigenvalue/singular values of a symmetric matrix $A$ are denoted as $\lambda_{\max}(A)$, $\lambda_{\min}(A)$, $\sigma_{\max}(A)$, and $\sigma_{\min}(A)$, respectively.
\edit{The spectral norm of a matrix $A \in \mathbb{R}^{m\times n}$ is denoted as $\|A\|_2$, which is also the maximum singular value of $A$.}
A zero vector with a one in the $i^{th}$ entry is denoted as $1_i$. 
Kronecker product is denoted as $\otimes$. 
A sequence of variables $\{x_i\}_{1:n}$ is concatenated into a vector through the operator $\text{vec}_{1:n}(x_i)$.

\section{Preliminary}\label{chp:prelim}
The concepts of model structure, regressor, information matrix, persistent excitation, and structural and practical identifiability are widely used across various research fields, including statistics, automatic control, systems biology, and environmental science \cite{johnson2002applied, ljung1994global, bellman1970structural, hamby1994review}. 
Although these analytical constructs share similar core concepts, the terminology and definitions may vary according to different conventions. 
In this section, we state the problem of interest and define these essential terms to ensure clarity and consistency. 

\subsection{Problem Formulation}

\edit{
Consider a continuous-time dynamical system:
\begin{equation}\label{dyn_constraint}
    \dot x(t) = \phi(x(t), u(t))\theta,
\end{equation}
where both input $u(t)$ and state $x(t)$ are measured at a finite set of sample time steps $\{t_0, \dots, t_K\}$ with equal distance $h = t_k - t_{k-1}$.
We define \(\phi_k \coloneqq \phi(x(t_k), u(t_k))\) and \(\psi_k \coloneqq (x(t_{k+1}) - x(t_k)) / h\).
Therefore, a stream of data points from \eqref{dyn_constraint} can be denoted as \(\{(\phi_k, \psi_k)\}_{0:K-1}\).}
Consider a sequence of pairs \(\{(\phi_k, \psi_k)\}_{0:K}\) that exhibits the noise characteristics:
\begin{equation}\label{eqn:gen_nonlinear_sys_pre}
    \psi_k = \phi_k\theta + \xi_k,
\end{equation}
where \(\phi_{k}\in \mathbb{R}^{n\times p}\) collects the feature vectors, \(\psi_{k}\in  \mathbb{R}^n\) is the target vectors, and \(\xi_k \in \mathbb{R}^n\) is \edit{a random} noise.
The residual of an estimator $\hat \theta$ is defined as:
\begin{equation}\label{eqn:residual}
    r_{k}(\hat{\theta}) = \psi_k - \phi_k\hat{\theta}.
\end{equation}
We want to design an adaptive identification law $f$ such that
$$\hat{\theta}_k \coloneqq \hat{\theta}_{k-1} + f(\phi_k, \psi_k),~\hat{\theta}_0 \coloneqq \theta_0,$$
{Interpreting $\theta_0$ in a Bayesian manner makes the role of this initialization explicit: assume a Gaussian prior distribution $\mathcal{N}(\theta_0, P_0),$ with covariance $P_0\succ 0$.}
 {With the weighted empirical cost:
 \begin{equation}\label{eqn:emp_cost_tseries}
 C_{K}^{(emp)}{(\hat \theta)}  \coloneqq  \sum_{i=1}^{K} w_{i,K}\,\|r_i(\hat \theta)\|^{2} +  \alpha^{K+1}\,\|\hat \theta-\theta_0\|_{P_0^{-1}}^{2},
\end{equation}
 we minimize $C_{K}^{(emp)}{(\hat \theta)}$ over $\hat \theta$ using all data up to time $K$, where $w_{i,K}$ is given in \eqref{eq:greedy_weights} and $\theta_0$ is the non-zero prior mean.
 We are interested in the online estimation of $\hat \theta$ under the lack of persistent excitation and poorly-conditioned practical identifiability (formally defined in Section~\ref{subsec:backgroud}).}
\edit{In Section~\ref{sec:GWRLS}, we show that choosing an appropriate form of the cost function can help mitigate the effect of the lack of persistent excitation and practical identifiability.
}
\edit{In Section~\ref{sec:fault-dect}, we develop an algorithm for}
$\hat{\theta}_k$ to track the time-varying $\theta(t)$ assuming that $\theta(t)$ is piece-wise constant in time.

\subsection{Background}\label{subsec:backgroud}
The notion of \textit{regressor signal} and \textit{regressor matrix} are defined as follows.

\begin{definition}[Regressor Signal]
    A regressor signal from time $0$ to $K$ is the sequence $\{\phi_k\}_{0:K}$.
    A regressor matrix $\Phi_{0:K}\in \mathbb{R}^{n(K+1)\times p}$ with respect to the regressor signal $\{\phi_k\}_{0:K}$ is:
    \begin{equation*}
        \Phi_{0:K} \coloneqq \begin{bmatrix}
            \phi_{0}^\top & \dots & \phi_K^\top
        \end{bmatrix}^\top.
    \end{equation*}
\end{definition}
The regressor signal constitutes the dataset of a learning problem, and arranging data points into rows within this signal gives rise to the regressor matrix, which goes by various names like \edit{the} design matrix\edit{,} or model matrix\edit{,} in different contexts.
Assessing the richness of the data naturally involves examining the rank of this regressor matrix. 
Yet, since the size of the regressor matrix varies depending on the quantity of available data points, it proves more convenient to analyze the \textit{information matrix} instead.
\begin{definition}[Information Matrix]
    The \textit{information matrix} with respect to $\{\phi_k\}_{0:K}$ is \edit{$H \coloneqq [\Phi^\top\Phi]_{0:K} \coloneqq \Phi_{0:K}^\top\Phi_{0:K} \in \mathbb{R}^{p\times p}$.}
\end{definition}
\noindent 
\edit{Note, by construction, $H$ is positive semi-definite.}
The information matrix \edit{$H$} can be understood through the lens of optimization theory. 
When the model parameters are linearly separated from the observed data, 
\edit{$H$} takes on the form of the Hessian matrix pertaining to the quadratic cost \edit{of} the model parameters. 
Essentially, it represents the curvature of the least squares cost function~\eqref{eqn:emp_cost_tseries}.
Also\edit{,} note that we can rewrite the information matrix as a series, $\Phi_{0:K}^\top\Phi_{0:K} = \sum_{k=0}^K \phi_k^\top\phi_k$, and the notion of persistent excitation can be understood as the persistent richness of future incoming data points, which is the rank of $\lim_{K\to \infty}\sum_{k=l}^K \phi_k^\top\phi_k$ for each $l > 0$.
\begin{definition}[Persistent Excitation~(PE)]
    A regressor signal $\lim_{K\to \infty}\{\phi_k\}_{0:K}$ is said to be persistently exciting if:
    \begin{equation}\label{eq:pe_definition}
        \sum_{k=l}^{l+L} \phi_k^\top \phi_k \succeq \edit{a} I \quad \forall l \in \mathbb{Z}_{\geq 0},
    \end{equation}
    for some positive constants $L, \edit{a > 0}$.
\end{definition}
Different definitions of persistent excitation have been proposed in the literature, as evidenced by various sources \cite{dhar2022initial, jha2019initial, panteley2001relaxed}, depending on whether the bounding condition applies universally to all initial conditions $x_0$ and inputs. 
Despite this variety, the essence of persistent excitation revolves around constraining the partial sum $\sum_{k=l}^{l+L} \phi_k^\top \phi_k$, which we will term the \textit{partial information matrix}. 

\edit{While \edit{the} lower bound in \eqref{eq:pe_definition} ensures the minimum eigenvalue of the information matrix is positive, 
a small condition number of the information matrix ensures that the matrix is well-conditioned.
}
\edit{A} well-conditioned information matrix is crucial for generating small confidence intervals in \edit{the} system parameter estimates. 
The collective requirements on model structure and signal richness that lead to these small confidence intervals are termed \textit{practical identifiability}.
\begin{definition}[Practical Identifiability~(PI) \edit{Index}]\label{def:pi_condtion}
    The practical identifiability index of $\{\phi_k\}_{k_1:k_2}$ is the condition number of the partial information matrix, i.e., 
    \begin{equation*}
    {
        \kappa_{k_1:k_2} \coloneqq \kappa([\Phi^\top\Phi]_{k_1:k_2}) \coloneqq \left\|[\Phi^\top\Phi]_{k_1:k_2}\right\|_{{2}}\cdot\left\|[\Phi^\top\Phi]_{k_1:k_2}^{-1}\right\|_{{2}}.}
    \end{equation*}
\end{definition}
The PI index $\kappa_{k_1:k_2} \in [1, \infty)$ indicates that the regressor signal is becoming ill-conditioned as its value approaches infinity.

\subsection{Lack of PE and PI in Dynamical Systems}\label{sec:lack-pe-pi}
The following theorem provides a sufficient condition for the lack of persistent excitation for a class of nonlinear systems.
\begin{proposition}\label{thm:general_lack_pe}
    \edit{Let there be a dynamical system with input and state pair $(u(\cdot), x(\cdot))$.
    A regressor signal $\phi(u(\cdot), x(\cdot))$ is not persistently exciting with respect to $(u(\cdot), x(\cdot))$ if:}
    \begin{enumerate}
        \item $u(t), x(t)$ converges to an equilibrium $u_*, x_*$ such that $(\phi_*^\top\phi_*)$ is rank deficient, and
        \item the feature map $\phi$ is continuous in \((u, x)\),
    \end{enumerate}
    \edit{where $\phi_* \coloneqq \phi(u_*, x_*)$.}
\end{proposition}
\begin{proof}
We first define the function $f: \mathbb{R}^{(\bar{n} + \bar{m}) \times (L+1)} \to \mathbb{R}$ as:
\begin{equation*}
    f(\mathbf{z}_{l:l+L}) \vcentcolon= \lambda_{\min}\left(\sum_{k=l}^{l+L}(\phi^\top\phi)\left(u(t_k), x(t_k)\right)\right),
\end{equation*}
where $\mathbf{z}_{l:l+L} \vcentcolon= [z(t_l), \dots, z(t_{l + L})]$, and $z(t)\coloneqq [u^\top(t), x^\top(t)]^\top$.
Let $\mathbf{z}_* \vcentcolon= [z_*, \dots, z_*] \in \mathbb{R}^{\bar{n}\times (L+1)}$.
To see that $f$ is a continuous function, we note the following. 
\begin{enumerate}
    \item Since $\phi(\cdot)$ is continuous in $u, x$ by definition, $\phi^\top(\cdot)$ is continuous, and $(\phi^\top\phi)(\cdot) \vcentcolon= \phi^\top(\cdot)\phi(\cdot)$ is also continuous by \cite[Theorem~4.9 and 4.10(a)]{rudin1964principles}.
    \item $\sum_{k=l}^{l+L}(\phi^\top\phi)(\cdot)$ is continuous by the same theorems.
    \item The eigenvalue function $\lambda_{\min}(\cdot)$ is continuous by \cite[Theorem~6.3.12]{horn2012matrix}.
\end{enumerate}
Therefore, the composite function $f$ is continuous.
Since $f$ is a continuous function, then by definition~\cite[Def. 4.5]{rudin1964principles}, for every $\epsilon > 0$, there exists a $\delta > 0$ such that:
\begin{equation}\label{eq:thm:general_lack_pe:proof:epsilon_bound}
    |f(\mathbf{z}_{l:l+L}) - f(\mathbf{z}_*)| < \epsilon
\end{equation}
for all $\mathbf{z}_{l:l+L}$ that satisfy $\|\mathbf{z}_{l:l+L} - \mathbf{z}_*\| < \delta$.

We want to show that for each $\epsilon$, there exists a $l_{\epsilon}$ such that $\|\mathbf{z}_{l:l+L} - \mathbf{z}_*\| < \delta$ for all data sizes $L$ and sample index $l \geq l_{\epsilon}$.
Recall that $z(t)$ converges to an equilibrium $z_*$ by the \edit{first} condition in the \edit{proposition} statement.
\edit{The convergence of $z(t)$} implies that every consecutive finite sub-sequence in the input-state trajectory, $\{z(t_l), \dots, z(t_{l+L})\} \subset \{z(t_{0}), \dots, z(t_{k}), \dots\}$,
is within some $\delta$-neighborhood of $z^{*}$. 
In other words, there exists a\edit{n} $l_{\epsilon} \geq 0$ \edit{such} that, if $l > l_{\epsilon}$, then there exists a $\delta>0$:
\begin{equation}\label{eq:thm:general_lack_pe:proof:delta_bound}
    \left\|\begin{bmatrix}
    \delta_l, \dots, \delta_{l+L}
    \end{bmatrix}^\top\right\|_{\infty} \leq \delta~\forall{L > 0},
\end{equation}
where $\{z_{*} + \delta_{l}, \dots, z_{*} + \delta_{l+L}\} = \{z(t_l), \dots, z(t_{l+L})\}.$

Therefore, by the continuity of $f$ and convergence of $z(t)$, 
for every $\epsilon > 0$, there exists a\edit{n} $l_\epsilon$ inferring a $\delta$-neighborhood around $z_*$ such that every finite consecutive sub-sequence \edit{is} bounded as stated in~\eqref{eq:thm:general_lack_pe:proof:delta_bound}. Consequently,  
from \eqref{eq:thm:general_lack_pe:proof:epsilon_bound}, we have:
\edit{\begin{equation*}
    |f(\mathbf{z}_{l:l+L}) - f(\mathbf{z}_*)| < \epsilon.
\end{equation*}
If $f(\mathbf{z}_{l:l+L})$ is smaller than $f(\mathbf{z}_*)$, then the corresponding sum of $\phi^T\phi$ is rank deficient and thus \eqref{thm:general_lack_pe:pf:step_03} holds. If it is larger than $f(\mathbf{z}_*)$, then we have}
\begin{align}
    \Rightarrow& f(\mathbf{z}_{l:l+L})) < f(\mathbf{z}_*) + \epsilon \label{thm:general_lack_pe:pf:step_01}\\
    \Rightarrow& \lambda_{\min}\left(\sum_{k=l}^{l+L} (\phi^\top \phi)(u(t_k), y(t_k))\right) < \epsilon \label{thm:general_lack_pe:pf:step_02}\\
    \Rightarrow& \sum_{k=l}^{l+L} \phi_k^\top \phi_k \not\succeq a I  \quad \forall l \geq l_{\epsilon} \label{thm:general_lack_pe:pf:step_03}
\end{align}
near the equilibrium \edit{$(u_*, x_*)$}. 
Note that \edit{the} step \edit{from} \eqref{thm:general_lack_pe:pf:step_01} \edit{to} \eqref{thm:general_lack_pe:pf:step_02} is possible because \edit{$(\phi_*^\top\phi_*)$} is rank deficient by assumption.
Additionally, the transition from \eqref{thm:general_lack_pe:pf:step_02} to \eqref{thm:general_lack_pe:pf:step_03} is justified since, for any $a > 0$, one can find a $0 < \epsilon < a$.
Hence, the positive definite condition of the partial information matrix is effectively violated around $(u_*, x_*)$.
\end{proof}


In an ideal scenario, 
when \edit{we are free to design the control input $u(t)$}, 
condition 1) in Proposition~\ref{thm:general_lack_pe} cannot hold, 
and one can ensure the rank of $\sum_{k=l}^{l+L} \phi_k^\top \phi_k$ by injecting a sinusoidal signal into $u(t)$.
However, this approach might not be possible when $u(\cdot) \equiv 0$, or when the cost of signal injection is too high.

Proposition~\ref{thm:general_lack_pe} provides a simple way to check the lack of PE.
Furthermore, Proposition~\ref{thm:general_lack_pe} sheds light on alternative approaches to addressing the issue of insufficient PE.
One approach involves model reduction to guarantee that the matrix $(\phi_*^\top\phi_*)$ is  full rank at the intended equilibrium.
The following example demonstrates the lack of excitation for a class of spreading systems with significant application value.

\begin{example}[\edit{SIS Epidemic Model Online Identification}]\label{ex:lack_pe_networked_SIS_SIR}
    The regressor signal of the susceptible-infected-susceptible~(SIS) epidemic model is not persistently exciting.
    We first define the model structure of the SIS dynamics:
    \begin{equation}\label{eq:simple_sis}
        \dot I(t) = (1 - I(t))\beta I(t) - \gamma I(t),
    \end{equation}
    where $I(t) \in [0, 1]$ is the infected proportion, $\beta$ is the infection rate, $\gamma$ is the recovery rate.
    Let $\psi_k \coloneqq \dot I(t_k)$, and $\phi_k \coloneqq \begin{bmatrix}
        (1 - I(t_k))I(t_k) & -I(t_k)
    \end{bmatrix}$, and $\theta \coloneqq \begin{bmatrix}
        \beta & \gamma
    \end{bmatrix},$
    where $I(t)$ is the state, corresponding to $x(t)$ in \eqref{dyn_constraint}\edit{. In this example the} 
    input signal $u(t) \equiv 0$ 
    and \edit{the derivative $\dot I(t)$ is the dependent variable in the online estimation problem.}
    Note that the regressor $\phi_k$ is continuous in $I$.
    If $I(t_0) = 0$, the lack of excitation is trivial; otherwise, if $I(t_0) > 0$, the system converges to an equilibrium $I_* = 1 - \gamma/\beta$ characterized by the system parameters~\cite{pare2020modeling} when $\beta > \gamma$.
    Furthermore, since $\phi_{*}^\top\phi_{*}$ is a rank-one matrix, $\{\phi_k\}$ is not persistently exciting with respect to~\eqref{eq:simple_sis} by Proposition~\ref{thm:general_lack_pe}.
    
    When the gradient descent algorithm in \cite{narendra2012stable} is applied to the online estimation problem:
    \begin{equation}\label{eq:sisExam:AdaIdLaw}
        \hat{\theta}_{k+1} = \hat{\theta}_k + \phi_k^\top\left(\psi_k - \phi_k \hat{\theta}_k\right),
    \end{equation}
    we can see that the $\hat{\theta}_k$ fails to converge to the true parameters in Fig.~\ref{fig:ex:contour}.
    \begin{figure}
        \centering
        \begin{subfigure}{\columnwidth}
            \includegraphics[width=\columnwidth]{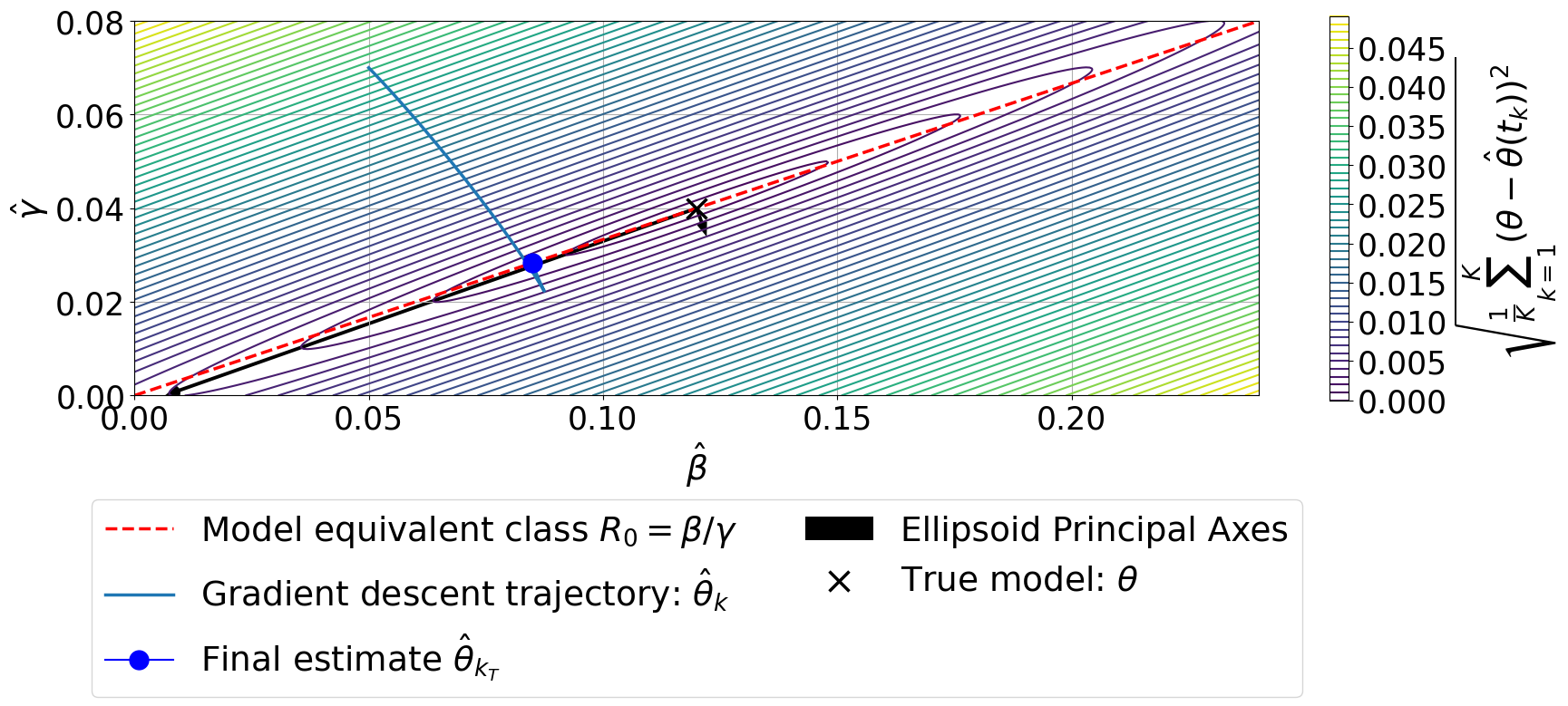}
            \caption{Error contour and estimation trajectory}
            \label{fig:ex:contour}
        \end{subfigure}
        \begin{subfigure}{\columnwidth}
            \includegraphics[width=\columnwidth]{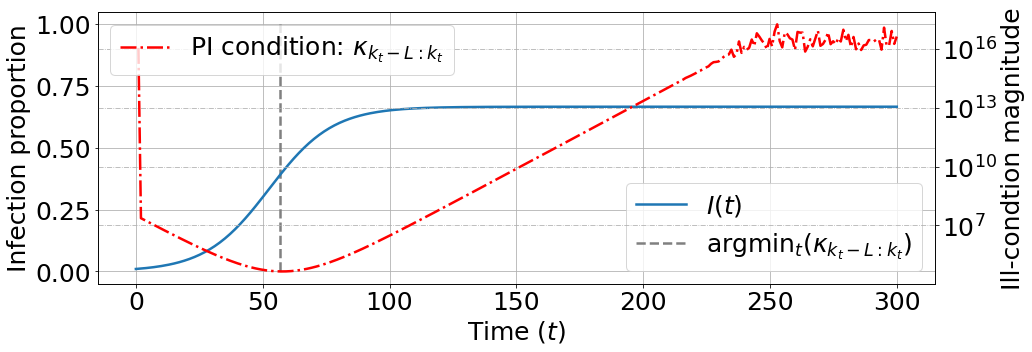}
            \caption{Measured state and PI condition}
            \label{fig:ex:traj}
        \end{subfigure}
        \caption{
            The trajectory of $\hat\theta_k$ on the RMSE contour plot is shown in Fig.~\ref{fig:ex:contour} for a noiseless SIS model $(\beta = 0.12$, $\gamma = 0.04$, $I(0) = 0.01$, $\hat{\beta}_{k_0} = 0.05$, $\hat{\gamma}_{k_0} = 0.07)$.
            The red dotted line in Fig.~\ref{fig:ex:contour} is the subspace of parameters characterized by the same reproduction number $\beta/\gamma$. 
            We choose the signal length $L=2$ for computing the PI condition in Fig.~\ref{fig:ex:traj}.
        }
        \label{fig:ex:error_main}
    \end{figure}
    Instead, the estimator $\hat{\theta}_k$ gravitates towards the model subspace defined by the reproduction number $\mathcal{R}_0 = \beta/\gamma > 1$, as denoted by the red dotted line in Fig.~\ref{fig:ex:contour}.
    This phenomenon implies an alternative reparameterization $(\tilde\theta, \tilde\phi, \tilde\psi)$ of \eqref{eq:simple_sis} that is robust against rank deficiency near $I_*$ is needed.

    In Fig.~\ref{fig:ex:traj}, the moving PI condition $\kappa_{k_t-L:k_t}$ as defined in Definition~\ref{def:pi_condtion}, with signal duration of $L=2$ attains its minimum at $t=57$.
    Subsequently, $\kappa_{k_t-L:k_t}$ increases monotonically until the information matrix loses its invertibility after $t=250$.
    This phenomenon suggests that the information carried by $\{\phi_k\}$ can be decomposed into the transient and the steady-state components.
    Furthermore, 
    transient information persists only for the initial period between $t=[0, 100]$ and reaches its peak at $t=57$, as we observe in Fig.~\ref{fig:ex:traj}.
    Once the system shifts from the transient phase to the steady-state, transient excitation is lost and the associated information matrix becomes ill-conditioned.
\end{example}

\section{Greedily-Weighted Recursive Least Squares}\label{sec:GWRLS}
Within the context of online parameter estimation, one of the issues that we face is the insufficient richness of the information provided by the input and state trajectory\edit{, as shown in the above example}.
Traditional solutions to the problem involve injecting exciting signals into the inputs. 
However, some other studies suggest that properly forgetting/remembering information in the past will help alleviate the diminishing richness of signal information over time\cite{goel2020recursive}.
This section delves into the challenges of diminishing information richness by analyzing the lack of persistent excitation and poor practical identifiability problems.
We propose a solution to the above problems based on the notion of an excitation set.
\subsection{Excitation Set}\label{sec:exc-set}
One natural solution to the lack of PE and PI problems in Example~\ref{ex:lack_pe_networked_SIS_SIR} is to incorporate second-order information when computing the descent direction upon the arrival of every new data point.
We introduce the following modification to the exponential forgetting recursive least squares~(EF-RLS) algorithm.
The proposed algorithm is detailed in Algorithm~\ref{alg:g_rls}, and we refer to it as the Greedily-Weighted Recursive Least Squares (\edit{GW-RLS}) Algorithm. We now introduce the notion of the \textit{optimal excitation set} and the \textit{greedy excitation set}
to illustrate the workings of the \edit{GW-RLS} Algorithm.

\begin{definition}[Optimal Excitation Set]\label{def:optimal_set}
    \edit{An optimal excitation set $\mathcal{E}_{\text{opt}}$ is a subset of data indices $\mathcal{E} \subset \{0, 1, \dots, K\} =: \mathcal{D}$
    such that: }
    \begin{equation*}
        \mathcal{E}_{\text{opt}}(K) = \underset{\substack{{\edit{\mathcal{E}}} \subseteq \mathcal{D}}}{\arg\min}\ \kappa\Big(\sum_{k \in \edit{\mathcal{E}}} \phi_k^\top \phi_k\Big).
    \end{equation*}
\end{definition}

Similarly, we define the greedy excitation set for a streamlined approximation of its optimal counterpart.

\begin{definition}[Greedy Excitation Set]\label{def:greedyset}
    The $K^{th}$ data point belongs to the greedy excitation set $\mathcal{E}_{\text{g}}(K)$ if it does not deteriorate the condition number \edit{of the information matrix} induced by $\mathcal{E}_{\text{g}}(K-1)$:
    \begin{equation*}
    \begin{small}
        \kappa\Big(\sum_{\substack{ k \in\mathcal{E}_{\text{g}}(K-1)\\ \bigcup \{K\}}} \phi_k^\top \phi_k\Big) \leq \kappa\Big(\sum_{k \in \mathcal{E}_{\text{g}}(K-1)} \phi_k^\top \phi_k\Big).
        \end{small}
    \end{equation*}
\end{definition}
\begin{remark}[Limitations of using the condition number]
    Selecting data by minimizing the condition number $\kappa(H)=\kappa\big(\sum_{k\in\mathcal{E}}\phi_k^\top\phi_k\big)$ helps avoid covariance windup but may obscure the requirement that the smallest singular value $\sigma_{\min}(H)$ grows~\cite{lai1982least}.
    In noisy settings, a strict condition-number rule may also prefer realizations with unusually small residuals, increasing the risk of overfitting.
    Practitioners should therefore monitor $\sigma_{\min}(H)$ alongside $\kappa(H)$ and adopt simple safeguards, e.g., occasionally accepting points that raise $\sigma_{\min}$.
    In this work, we focus on the study of condition-number–based greedy selection within an EF-RLS framework.
    A formal study of greedy selection policies that directly control $\sigma_{\min}$ is left to future work.
\end{remark}

\subsection{Excitation Set Based Recursive Least Squares}\label{sec:exc-set-rls}
The main modification of Algorithm~\ref{alg:g_rls} in the EF-RLS filtering lies in its strategy to preserve data points from the optimally exciting set, ensuring they remain undiluted by subsequent less-informative data.
Since solving for the optimal excitation set at each update iteration is expensive, we propose a recursively approach to approximate the optimally exciting set through a greedy excitation set.
The pseudo-code for \edit{GW-RLS} is provided below.

\begin{algorithm}
    \caption{Greedily-Weighted Recursive Least Squares}\label{alg:g_rls}
    \begin{algorithmic}[1]
        \Initialize{
            $H^{(e)}_0$, $P_0$, $\theta_0$, $\Phi^{(e)}_0$, $\upsilon^{(e)}_0$,
            \edit{$\alpha$}.
        }
        \Inputs{\label{alg:g_rls:line:input}
            Input datum $(\phi_{k}, \psi_{k})$
        }
        \If{$k > 0$}
            \State $H^{(e)} \gets H^{(e)}_{k-1} + \phi_k^\top\phi_k$\label{alg:g_rls:line:greedy_begin}
            \If{$\kappa(H^{(e)}) \leq \kappa{(H^{(e)}_{k-1})}$}\label{alg:g_rls:line:best_condition_num_sf}
                \State $\Phi^{(e)}_{k}, H^{(e)}_{k} \gets [{\Phi^{(e)}_{k-1}}^\top, \phi_k^\top]^\top, H^{(e)}$\label{alg:g_rls:line:define_Phi_e_in_greedy}
                \State $\upsilon^{(e)}_{k} \gets \upsilon^{(e)}_{k-1} + \phi_k^\top\psi_k$\label{alg:g_rls:line:define_upsilon_e_in_greedy}
                \State $\Phi \gets (\sqrt{1-\alpha})\Phi^{(e)}_{k}$\label{alg:g_rls:line:define_Phi_in_greedy}
                \State $H, \upsilon \gets (1-\alpha)H^{(e)}_{k}, (1-\alpha)\upsilon^{(e)}_{k}$\label{alg:g_rls:line:define_var_in_greedy}
            \Else
                \State $\Phi^{(e)}_{k}, H^{(e)}_{k}, \upsilon^{(e)}_{k} \gets \Phi^{(e)}_{k-1}, H^{(e)}_{k-1}, \upsilon^{(e)}_{k-1}$\label{alg:g_rls:line:define_var_out_greedy}
                \State $\Phi \gets [(\sqrt{1-\alpha}){\Phi^{(e)}_{k}}^\top, \phi_k^\top]^\top$\label{alg:g_rls:line:define_Phi_out_greedy}
                \State $H \gets (1-\alpha)H^{(e)}_{k} + \phi_k^\top\phi_k$\label{alg:g_rls:line:define_H_out_greedy}
                \State $\upsilon \gets (1-\alpha)\upsilon^{(e)}_{k} + \phi_k^\top\psi_k$\label{alg:g_rls:line:define_upsilon_out_greedy}
            \EndIf\label{alg:g_rls:line:greedy_end}
            \State $P_{k} \gets \frac{1}{\alpha}P_{k-1} (I - {\Phi}^\top (\alpha I + \Phi P_{k-1} \Phi^\top)^{-1} \Phi P_{k-1})$\label{alg:g_rls:line:P_update}
            \State $\hat{\theta}_{k} \gets \hat{\theta}_{k-1} + P_{k}(\upsilon - H\hat{\theta}_{k-1})$\label{alg:g_rls:line:theta_update}
        \EndIf
    \end{algorithmic}
\end{algorithm}

In words, the \edit{GW-RLS} Algorithm can be summarized as follows:  \edit{upon receiving a new data point (Line~\ref{alg:g_rls:line:input}), we check to see if that data point improves the condition number of the information matrix (Line \ref{alg:g_rls:line:best_condition_num_sf}).}
If $\kappa(H^{(e)}) \leq \kappa{(H^{(e)}_{k-1})}$, then we \edit{add the new data point to the excitation set and} update the regressor matrix $\Phi^{(e)}_{k}$, the information matrix $H^{(e)}_{k}$\edit{,} and the corrector term $\upsilon^{(e)}_{k}$ of the excitation set (Lines~\ref{alg:g_rls:line:define_Phi_e_in_greedy}-\ref{alg:g_rls:line:define_upsilon_e_in_greedy}); else, $\Phi^{(e)}_{k}, H^{(e)}_{k}, \upsilon^{(e)}_{k}$ remain unchanged \edit{(Line~\ref{alg:g_rls:line:define_var_out_greedy})}.
Lastly, \edit{we} use $\Phi^{(e)}_{k}, H^{(e)}_{k}, \upsilon^{(e)}_{k}$ to \edit{update} $\Phi, H, \upsilon$  \edit{(Lines \ref{alg:g_rls:line:define_Phi_in_greedy}-\ref{alg:g_rls:line:define_var_in_greedy}, \ref{alg:g_rls:line:define_Phi_out_greedy}-\ref{alg:g_rls:line:define_upsilon_out_greedy}),} which are \edit{then employed to compute} the inverse information matrix $P_{k}$ \edit{(Line~\ref{alg:g_rls:line:P_update})} and the new estimate $\hat{\theta}_{k}$ \edit{(Line~\ref{alg:g_rls:line:theta_update})}.


\edit{\begin{table}
    \centering
    \renewcommand{\tablename}{\centering \textbf{Table}} 
    \resizebox{0.47\textwidth}{!}{%
    \begin{tabular}{|c|l|l|}
        \hline
        \textbf{Variable} & \textbf{\centering Meaning} & \textbf{Initialization Values} \\
        \hline
        $H^{(e)}_0$ & Information matrix of the excitation set & $0$ \\
        $P_0$ & Prior covariance matrix & $P_0 \in \{\rho I : 10^3 \leq \rho \leq 10^6\}$ \\
        $\theta_0$ & Prior parametric mean & Any admissible value of $\theta$ \\
        $\Phi^{(e)}_0$ & Regressor matrix of excitation set & Empty sequence \\
        $\upsilon^{(e)}_0$ & Cross-correlation vector of excitation set & $0$ \\
        $\alpha$ & Exponential decaying rate & $0.95 \leq \alpha \leq 0.9999$ \\
        \hline
    \end{tabular}%
    }
    \caption{Initialization variables and their meanings}
    \label{tab:init_values}
\end{table}
}

\edit{Table~\ref{tab:init_values} outlines the meaning and values of the set of initial variables. We define $\kappa(0) = \infty$, so the first data point will be naturally accepted into the excitation set. Note that 
$\alpha, H^{(e)}_0$, $P_0$, $\hat{\theta}_0$, $\Phi^{(e)}_0$, and $\upsilon^{(e)}_0$ are the only variables that need to be initialized.}
For setting the values of $P_0$ and $\alpha$, we follow standard practice in the literature~\cite{ljung1995system}.
\edit{Furthermore, the choice of $\alpha$ and $P_0$ are quite flexible in the \edit{GW-RLS} framework.
Typical values used in EF-RLS are viable for \edit{GW-RLS}.
We set $P_0 = 10^5I$ and $\alpha=0.98$ in the simulations unless otherwise specified.
}
\begin{remark}[Well-conditioned Information Matrix]
    \edit{
    Notice that the maximum eigenvalue, $\lambda^{(e)}_{\max}(k)$, of $H^{(e)}_k$ is asymptotically bounded by its minimum eigenvalue, $\lambda^{(e)}_{\min}(k)$, under some regularity condition on the regressor signal.
    If the regressor signal $\{\phi_k\}$ satisfies the following inequality:}
    \begin{equation*}
        bI \succ \sum_{k=0}^L \phi_k^\top \phi_k \succ aI,
    \end{equation*}
    \edit{for some positive constants $L, a, b \in \mathbb{R}_{>0}$,
    which states that an initial interval of $\{\phi_k\}$ is well-conditioned,
    then, we have:}
    \begin{align*}
        \kappa(H_k) &\leq \frac{\lambda^{(e)}_{\max}(L)}{\lambda^{(e)}_{\min}(L)} < \frac{b}{a}\\
        \lambda^{(e)}_{\max}(k) &< (b/a) \lambda^{(e)}_{\min}(k), \qquad \forall k \geq L,
    \end{align*}
    \edit{since $\mathcal{E}_g(n)$ is constructed in a way that the condition number of its associated information matrix is non-increasing, as stated in Definition~\ref{def:greedyset}.
    The conditioning of $H^{(e)}_k$ guarantees that $H$ is well-conditioned at each update step.}
\end{remark}
Furthermore, the condition number $\kappa({H^{(e)}})$ can be recursively estimated, as detailed in \cite{benesty2005recursive}, to fully streamline Algorithm~\ref{alg:g_rls}.
We are \edit{now} ready to characterize the optimality of Algorithm~\ref{alg:g_rls}.
\begin{theorem}\label{thm:opt_theta}
    If $P_0$ is positive definite, 
    then for all $k \in [\edit{1}, \infty)$, $\hat{\theta}_{k+1}$ obtained by Algorithm~\ref{alg:g_rls} is the the unique minimizer of the cost function:
    \begin{align}\label{eq:cost}
        C_{k}\edit{(\theta)} = \sum_{i=\edit{1}}^{k}w_{i,k}\|r_{i}(\edit{\theta})\|^2 + \alpha^{k+1}\|\edit{\theta} - \theta_0\|_{P_0^{-1}}^2
    \end{align}
    with the weighting function defined as:
    \begin{equation}\label{eq:greedy_weights}
        w_{i,k} =
        \begin{cases}
            (1-\alpha)\sum_{l=i}^{k}\alpha^{k-l} & \text{if $i \in \mathcal{E}_{\text{g}}(k)$,} \\
            \alpha^{k-i}               & \text{otherwise}.
        \end{cases}
    \end{equation}
\end{theorem}
\begin{proof}
    The proof leverages mathematical induction to proceed, and it suffices to show the inductive step.
    We first note that $C_{k}(\hat{\theta})$ can be written in terms of:
        $C_{k}(\hat{\theta}) = \hat{\theta}^\top A_k\hat{\theta} + 2b_k^\top \hat{\theta} + c_k$,
    where $A_k, b_k, c_k$ are:
        \begin{align*}
            A_k & := \sum_{i=0}^k w_{i, k} \phi_i^\top \phi_i + \alpha^{k+1}P^{-1}_0                 \\
            b_k & := -\sum_{i=0}^k w_{i, k} \phi_i^\top \psi_i - \alpha^{k+1}P^{-1}_0\theta_0            \\
            c_k & := \sum_{i=0}^k w_{i, k} \psi_i^\top \psi_i + \alpha^{k+1}\theta_0^\top P^{-1}_0\theta_0.
        \end{align*}
    Then, $A_k, b_k$ can be computed recursively as
        \begin{align*}
        A_k & = \alpha A_{k-1} +\sum_{i\in U}\omega_i \phi_i^\top\phi_i \\
        b_k & = \alpha b_{k-1} - \sum_{i\in U}\omega_i \phi_i^\top \psi_i,
    \end{align*}
    where $U = \mathcal{E}_{\text{g}}(k) \bigcup \{k\}$, and
    \edit{
    $$
        \omega_i = \begin{cases}
            (1 - \alpha) & \text{if $i \in \mathcal{E}_{\text{g}}(k)$}, \\
            1 & \text{otherwise}.
        \end{cases}
    $$}
    By way of induction, assume $\exists~k\in \mathbb{N}$ such that $A_{k-1}$ is positive definite and the unique optimizer of $C_{k-1}$ is $\hat{\theta}_{k} = -A_{k-1}^{-1}b_{k-1}$.
    We define $P_{k+1} := A_k^{-1}$. 
    Since $A_{k-1}$ is positive definite,
    we can apply the matrix inversion lemma\cite[p, 304]{bernstein2009matrix} and obtain a positive definite $P_{k+1}$:
    \begin{align*}
        P_{k+1} & = A_k^{-1}\\
        & = \frac{1}{\alpha}\left(A_{k-1} +\frac{1}{\alpha}\left(\sum_{i\in U}\omega_i \phi_i^\top\phi_i\right)\right)^{-1}\\
        &= \frac{1}{\alpha}P_k - \frac{1}{\alpha}P_k{\Phi}^\top (\alpha I + \Phi P_k \Phi^\top)^{-1} \Phi P_k,
    \end{align*}
    where 
    \edit{
    $$
        \Phi = \begin{cases}
            (\sqrt{1 - \alpha})\Phi^{(e)}_{k+1} &\text{if $k \in \mathcal{E}_{\text{g}}(k)$,} \\
            [(\sqrt{1 - \alpha})(\Phi^{(e)}_{k+1})^\top, \phi_k^\top]^\top & \text{otherwise},
        \end{cases}
    $$}and $\Phi^{(e)}_{k+1} = \begin{bmatrix}
                \phi_{k_1}^\top, \phi_{k_2}^\top, \dots, \phi_{k_r}^\top
            \end{bmatrix}_{k_i\in \mathcal{E}_{\text{g}}(k)}^\top$, which satisfies the computation of $\Phi^{(e)}_{k+1}$ and $\Phi$ on Lines~\ref{alg:g_rls:line:define_Phi_e_in_greedy}, \ref{alg:g_rls:line:define_Phi_in_greedy}, \ref{alg:g_rls:line:define_var_out_greedy}, and \ref{alg:g_rls:line:define_Phi_out_greedy} of Algorithm~\ref{alg:g_rls}.
    By the quadratic minimization lemma~\cite{bernstein2009matrix}, the unique minimizer of $C_{k}(\hat{\theta})$ is:
    \begin{small}
    \begin{align*}
        \hat{\theta}_{k+1} & = -A_k^{-1}b_k                           \\
                     & = A_k^{-1}\Big(-\alpha b_{k-1} + \sum_{i \in U}\omega_i\phi_i^\top \psi_i\Big)\\
                     & = A_k^{-1}\Big(\alpha A_{k-1} \hat{\theta}_k + \sum_{i \in U}\omega_i\phi_i^\top \psi_i\Big)\\
                     & = A_k^{-1}\Big(\Big(A_k - \sum_{i \in U}\omega_i\phi_i^\top\phi_i\Big)\hat{\theta}_k + \sum_{i \in U}\omega_i\phi_i^\top \psi_i\Big)\\
                     & = \hat{\theta}_k + A_k^{-1}\Big(\sum_{i \in U}\omega_i\phi_i^\top \psi_i - \Big(\sum_{i \in U}\omega_i\phi_i^\top\phi_i\Big)\hat{\theta}_k \Big)\\
                     & = \hat{\theta}_k + P_{k+1}(\upsilon - H\hat{\theta}_k),
    \end{align*}
    \end{small}
    
    \noindent 
    where $\upsilon := \sum_{i \in U}\omega_i\phi_i^\top \psi_i$ and $H := \sum_{i \in U}\omega_i\phi_i^\top\phi_i$ matches the computation of $\upsilon$ and $H$ on Lines~\ref{alg:g_rls:line:define_var_in_greedy}, 
    \ref{alg:g_rls:line:define_H_out_greedy}, and
    \ref{alg:g_rls:line:define_upsilon_out_greedy}
    of Algorithm~\ref{alg:g_rls}.
    By the principle of mathematical induction, $P_n$ is positive definite, and $\hat{\theta}_{n+1}$ is the unique minimizer of $C_n$ for all $n \in \mathbb{N}$.
\end{proof}

While Theorem~\ref{thm:opt_theta} characterizes the cost function \edit{that} Algorithm~\ref{alg:g_rls} optimizes, the structure of the weighting function $w_{i,k}$ might not be immediately obvious. 
The intuition of the choice of $w_{i,k}$ becomes apparent when considering its asymptotic behavior, as established below.

\begin{corollary}\label{cor:limit_cost}
    If $P_0$ is positive definite and all data \edit{with index} $\edit{i}\in \lim_{k\to \infty}\mathcal{E}_g(k)$ are obtained before $k = K$, for some $K < \infty$, and $\alpha < 1$, 
    then, as $k \to \infty$, $\hat{\theta}_{k+1}$ obtained by Algorithm~\ref{alg:g_rls} is the unique minimizer of the cost function:
    \begin{align}\label{eq:cost2}
        C_{k}(\hat{\theta}) = \sum_{i=\edit{1}}^{k}W_{i,k}\|r_i(\hat{\theta})\|^2,
    \end{align}
    with
    \begin{equation}\label{eq:cor:limit_weight}
        W_{i,k} = \begin{cases}
            1 & \text{if $i \in 
            \mathcal{E}_g(K)$}\\
            \alpha^{k-i} & \text{otherwise}.
        \end{cases}
    \end{equation}
\end{corollary}
\begin{proof}
    Note that, since $\alpha<1$, the second term in the cost function in \eqref{eq:cost}, $\alpha^{k+1}\|\hat{\theta} - \theta_0\|_{P_0^{-1}}^2$, goes to zero as $k \to \infty$.
    Furthermore, $\sum_{l=i}^{k}\alpha^{k-l}$ can be rewritten as 
    $\sum_{l=0}^{k-i}\alpha^{l}$,
    and since $\sum_{l=0}^{k-i}\alpha^{l}$ is a geometric sum,
    $\lim_{k\to \infty} (1-\alpha)\sum_{l=i}^{k}\alpha^{k-l} = 1$.
    Thus, the weighting function $w_{i,k}$ can be written as $W_{i,k}$ in \eqref{eq:cor:limit_weight}.  
    Therefore, by Theorem~\ref{thm:opt_theta}, Algorithm~\ref{alg:g_rls} obtains the optimal $\theta^* = {\arg\min}_{\hat{\theta}} C_{k}(\hat{\theta})$ for the simplified cost function in \eqref{eq:cost2}-\eqref{eq:cor:limit_weight} \edit{as $k\to \infty$}.
\end{proof}
Corollary~\ref{cor:limit_cost} characterizes the asymptotic behavior of the cost function, which Algorithm~\ref{alg:g_rls} minimizes, \edit{when no data points are added to the excitation set asymptotically.}
The algorithm stores the data points that belong to the greedy excitation set by giving them a weight of one while assigning the unexciting points exponentially decaying weights.

\section{Hyper-Parameters Resetting with Change Point Detection}\label{sec:fault-dect}
\edit{In the last section, we have investigated a way to overcome the lack of persistent excitation in online identification.
However, GW-RLS and other online estimators that rely on selective forgetting processes sacrifice adaptability for stability.
}
In this section, we \edit{propose a remedy to improve the adaptability of GW-RLS and potentially the wider class of online estimation algorithms.}
\edit{We do so by considering a memory resetting scheme for online identification process operating on} time-varying systems with piecewise constant parameters in time.
\edit{The memory resetting scheme implements} a change point detection algorithm that detects jumps in $\theta(t_k)$ when the predictive ability of the estimated parameter vector $\hat{\theta}_k$ deteriorates.
Once a jump is detected, the estimator memory is reinitialized, 
in order to avoid estimation errors \edit{generated by} outdated data points from irrelevant parameters in the past.

\subsection{Prediction Error Dynamics}
Recall the prediction residual $e_k:=\psi_k-\phi_k\hat{\theta}_{k-1}$ and its squared norm $\|e_k\|^2$ (prediction error), which quantifies the predictability of $\hat\theta_{k-1}$; $e_k$ is used instead of $r_k$ in~\eqref{eqn:residual} to emphasize the use of time-varying parameters.
We can see that the prediction residual $e_k$ is a stochastic process with 
time-varying means and variances.
\edit{In other words,} $\|e_k\|^2$ is a signal dependent on the past data $\Phi_k$, the true parameters $\theta(t_{k})$, and some hyper-parameters specific to an online estimator \edit{such as $\theta_0$ and $W_k$}. 
Additionally, 
since the prediction error spans several orders of magnitude, 
we apply a negative logarithmic transformation to $\|e_k\|^2$ to define the notion of model predictability.
\begin{definition}
    The \textit{model predictability} at time $t_k$ is defined as $Y_k \coloneqq -\log(\|e_k\|^2)$, where $e_k\coloneqq \psi_k - \phi_k\hat{\theta}_{k-1}$.
\end{definition}
The \textit{model predictability} $Y_k\in (-\infty, \infty)$ is the negative order of magnitude of the prediction error.
While its value may depend on the application, a high $Y_k$ indicates good performance of the upstream online estimator.

\edit{We aim to create a heuristic that leverages the conditional expectation of $Y_k$ given $Y_{k-1}$ to detect the occurrence of a change between time $t_{k-1}$ and $t_k$.}
\edit{To achieve this goal, we employ the}
exponentially weighted moving average (EWMA)~\cite{hunter1986exponentially}: 
\begin{equation}\label{eq:EWMA}
    Z_k = \eta Y_k + (1 - \eta) Z_{k-1},~Z_0 \coloneqq Y_1,
\end{equation}
where $\eta \in [0, 1]$ is an exponential weighting term.
We can use the $Z_{k-\edit{1}}$ as a predictor of \edit{the expectation of $Y_k$}.
Therefore, the estimator of $Y_k$ is computed as:
$$
    \hat{Y}_k = Z_{k-\edit{1}}.
$$
The EWMA estimator effectively acts as a low pass filter on $\mathbf{Y} \coloneqq \{Y_k\}$. 
If $Y_k$ deviates too much from $\hat{Y}_k$, it is a good sign of a change point occurrence.
The following example demonstrates the interplay between a time-varying system, model predictability, and the EWMA estimator.
\begin{example}\label{ex:sir_cp}
    Consider a 2-node susceptible-infected-recovered~(SIR) system:
    \begin{small}
        \begin{equation}\label{eq:sir_cp}
            \begin{aligned}
                \dot I(t) &= \text{diag}(1 - I(t) - R(t))B(t)I(t) - \text{diag}(\gamma(t))I(t)\\
                \dot R(t) &= \text{diag}(\gamma(t))I(t),
            \end{aligned}
        \end{equation}
    \end{small}
    where $I \in [0, 1]^2$ is the infected proportion, $R \in [0, 1]^2$ is the recovered proportion, $B \in\mathbb{R}_{\geq 0}^2\times \mathbb{R}^2_{\geq 0}$ is the infection rate matrix, and $\gamma \in\mathbb{R}_{\geq 0}^2$ is the vector of recovery rates.
    The system parameters $B$ and $\gamma$ are time-varying but remain piece-wise constant, emulating the lockdown effects and rebound in social activity during the reopening period.
    \begin{figure}
        \centering
        \includegraphics[width=\columnwidth]{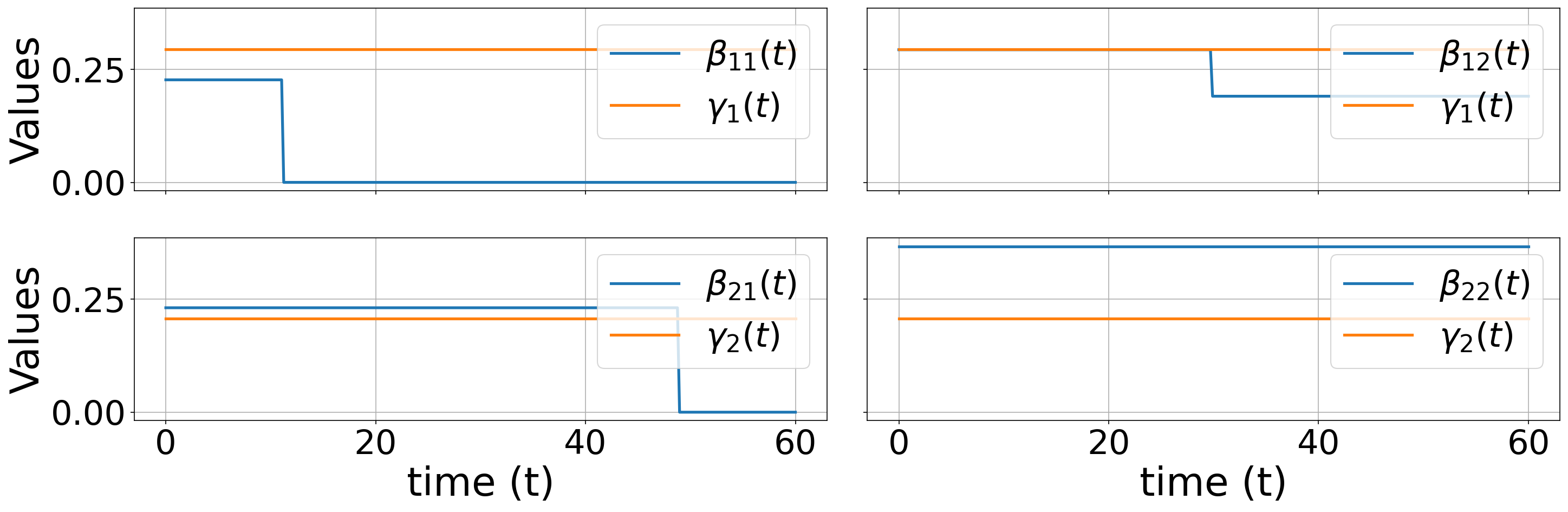}
        \caption{Time-varying parameters of a 2-node SIR network}
        \label{fig:2node_sir}
    \end{figure}
    Fig.~\ref{fig:2node_sir} visualizes the time-varying parameters.
    Let $R, I$, and their derivatives be observable; we can fit the data through an EF-RLS algorithm~\cite{islam2019recursive}.
    A visualization of the evolution of ${Y_k}$ in conjunction with the EWMA estimator $Z_{k-1}$ is shown in Fig.~\ref{fig:predictability_EWMA}.
    \begin{figure}
        \centering
        \includegraphics[width=\columnwidth]{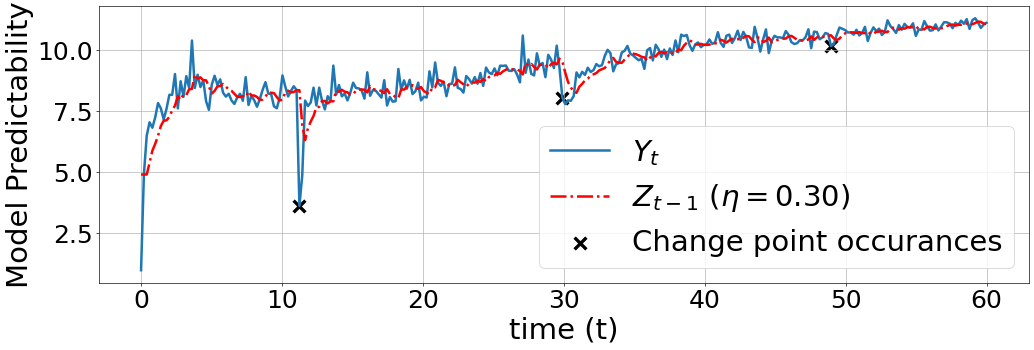}
        \caption{Evolving model predictability of a \edit{GW}-RLS and its EWMA ($\edit{\eta}=0.5$) estimation.}
        \label{fig:predictability_EWMA}
    \end{figure}
    \edit{It is evident that each change point occurrence is associated with a downward jump of $Y_t$, and the magnitude of the jump depends on the amount of change in $\theta_t$ and how that change is reflected in the observed states.
    }
\end{example}

\subsection{Change Point Detection}
In this subsection, we design a simple recursive likelihood ratio test to detect change point occurrences.
We construct two hypotheses 
to test for change points in $\theta(t)$:
\begin{itemize}
    \item[] $\bar{H}_0$: 
    The stochastic process $\{Y_i\}_{k_1:k_2}$ is ``captured'' by a single EWMA.
    \item[] $\bar{H}_1$: 
    The stochastic process $\{Y_i\}_{k_1:k_2}$ is ``captured'' by two EWMAs, residing in two intervals $[k_1,k_c]$ and $[k_c+1,k_2]$,
\end{itemize}
where $k_1 \leq k_c \leq k_2$.
In the context of our change point detection problem, it is advantageous to set $k_c=k_2-1$ because we want to prevent the most recent abnormal observation from impacting the upstream estimation process and EWMA value.

To clarify the meaning of ``captured,'' we need to establish a criterion that defines the cost of $\hat{Y}_k$ failing to predict $Y_k$ in the context of the detection problem.
Let $\mathcal{S}_{k_1:k_2} \coloneqq \{k_1, \dots, k_2\}$ be the index set of all residuals.
The set $\mathcal{S}_{k_1:k_2}$ is then partitioned into two residual index sets: the positive, $\mathcal{S}_{k_1:k_2}^+ \coloneqq \{k:\hat{Y}_k - Y_k > 0, k_1\leq k\leq k_2\}$, and non-positive, $\mathcal{S}_{k_1:k_2}^- = \{k:\hat{Y}_k - Y_k \leq 0, k_1\leq k\leq k_2\}$.
After fixing an exponential weight $\eta$, the cost of $\hat{Y}_k$ failing to predict is defined below.
\begin{definition}
    The mean positive squared error~(MPSE) is defined as the conditional expectation of $\hat{Y}_k - Y_k$ on the set $\mathcal{S}^+$:
    \begin{equation*}
        \text{MPSE}(\mathbf{Y}) = \mathbb{E}[(\hat{Y}_k - Y_k)^2 | k \in \mathcal{S}^+].
    \end{equation*}
\end{definition}
When there is an abrupt change, it results in $Y_k < \hat Y_k$ and an increase in $\|e_k\|^2$, consistent with the objective of our task.
However, note that the distribution of $\hat{Y}_k - Y_k$ after the partition can no longer assumed to be Gaussian, and a generalized error distribution should be considered when computing the likelihood.
To simplify \edit{the} computation of the likelihood, we make the following assumption on the distribution of $\hat{Y}_k - Y_k$.
\begin{assumption}\label{asm:exp_dist_res}
    Let $E_k \coloneqq (\hat{Y}_k - Y_k)^2$, where $k \in \mathcal{S}^+$, be 
    \edit{an independent and identically distributed~(i.i.d.) exponential random variable}, that is, $E_k|\mathcal{S}^+ \sim \text{Exp}(\lambda)$.
\end{assumption}
\noindent 
\edit{
Though seemingly restrictive, the $E_k$ assumption provides a simple heuristic to detect change points and reset the excitation set in Algorithm~\ref{alg:g_rls}.
}
With a slight abuse of notation, we define 
$\text{MPSE}(\{E_{i}\}_{k_1:k_2}) \coloneqq \mathbb{E}[E_k | k \in \mathcal{S}_{k_1:k_2}^+].$
Then, under Assumption~\ref{asm:exp_dist_res}, the null and alternative hypotheses for the likelihood ratio test can be constructed as follows:
\begin{itemize}
    \item[] $H_0$: \(E_i|\mathcal{S}^+ \sim \text{Exp}(\lambda_0)~\forall k_1\leq i\leq k_2\)
    \item[] $H_1$: \(E_i|\mathcal{S}^+ \sim \text{Exp}(\lambda_1)~\forall k_1\leq i < k_2\)\\ 
    and $E_{k_2}|\mathcal{S}^+ \sim \text{Exp}(\lambda_2)$
\end{itemize}
where $\lambda_0^{-1} = \text{MPSE}(\{E_i\}_{k_1:k_2})$ and $\lambda_1^{-1} = \text{MPSE}(\{E_i\}_{k_1:k_2-1})$ and are estimated by computing the sample mean of $E_k|\mathcal{S}^+_{k_1:k_2}$ and $E_k|\mathcal{S}^+_{k_1:k_2-1}$, respectively.
However, note that since $\mathcal{S}^+_{k_2:k_2}$ is empty, the statement ``$E_{k_2}|\mathcal{S}^+ \sim \text{Exp}(\lambda_2)$,'' is vacuously true regardless of the value of $\lambda_2$.
Therefore, the alternative hypothesis can be rewritten as $H_1$: \(E_i|\mathcal{S}^+ \sim \text{Exp}(\lambda_1)~\forall k_1\leq i\leq k_2-1\).

\begin{definition}[Likelihood Ratio Test Statistic]\label{def:lrt_statistic}
Let \( \mathbf{E} = \{E_1, E_2, \ldots, E_n\} \) be random samples from a probability distribution with parameters \(\Lambda\). 
Suppose \(\Lambda_0 = \edit{\{\lambda_0\}}\) is the set of the parameters under the null hypothesis, and let \(\Lambda_1 = \{\lambda_1^1, \edit{\lambda_1^{2}}\}\) be the parameters set under the alternative hypothesis. 
The test statistic \(D(\textbf{E})\) is defined as:
\begin{equation}\label{eq:lrt_stat}
    D(\textbf{E}) = -2 \ln \left( \frac{L(\textbf{E}|\Lambda_0)}{L(\textbf{E}|\Lambda_1)} \right)
\end{equation}
where \( L(\textbf{E}|\Lambda) \) is the likelihood function.
\end{definition}

In the context of our problem, the null hypothesis \(H_0\) posits that the sequence $\{Y_i\}_{k_1:k_2}$ is represented by only one EWMA, leading to a single exponential distribution $\exp(\lambda_0)$.
Similarly, \edit{$\Lambda_1 = \{\lambda_1^1, \lambda_1^{2}\}$} as a result of the alternative hypothesis \(H_2\) that posits that $\{E_k\}_{k_1:k_2}$ is generated by two \edit{different} exponential distributions.
Then, with the choice of the test statistic $D$ induced from Assumption~\ref{asm:exp_dist_res}, we can recursively perform the likelihood ratio test (LRT) with threshold value $\tau$ \edit{as elaborated in Algorithm~\ref{alg:rec_lrt_cp_dtect}}.
\begin{algorithm}
    \caption{Recursive LRT for change point detection~(rLRT)}\label{alg:rec_lrt_cp_dtect}
    \begin{algorithmic}[1]
        \Initialize{
            Error threshold: $\tau \in [0, 1]$\\
            Exponential weighting factor: $\eta \in [0, 1]$\\
            Initial EWMA: $Z_0 \gets Y_1$\\
            Initial residual mean: $\lambda_0 \gets 0$\\
            Number of downward drifts: $n \gets 0$
        }
        \vspace{1mm}
        \Inputs{
            Input datum: $Y_{k}$, Data count: $k$
        }
        \If{$k \geq 1$}
        \If{$Z_{k-1} - Y_k > 0$}
            \State $E \gets (Z_{k-1} - Y_k)^2$
            \State $\lambda \gets (n+1)/(\frac{n}{\lambda_n}+E)$
            \State $D \gets 2\left(n\ln(\lambda_n) - (n+1)\ln(\lambda) + 1\right)$
            \State $p \gets 1-\text{CDF}_{\chi^2}(D, 1)$
            \If{$p \leq \tau$}
                \Return True
            \EndIf
            \State $n \gets n + 1$
            \State $\lambda_n \gets \lambda$
        \EndIf
        \State $Z_k \gets \eta Y_k + (1 - \eta) Z_{k-1}$
        \State \textbf{Return} False
        \EndIf
    \end{algorithmic}
\end{algorithm}

In Algorithm~\ref{alg:rec_lrt_cp_dtect}, 
$D = 2\left(n\ln(\lambda_n) - (n+1)\ln(\lambda) + 1\right)$ 
is the aforementioned test statistic.
The properties of the hypothesis test follow naturally from the Wilks' Theorem~\cite{wilks1938large} and Neyman-Pearson Lemma~\cite{hogg1977probability}, reproduced below for completeness.
\begin{lemma}[Wilks' Theorem~{{\cite{wilks1938large}}}]\label{lem:wilks}
Under the null hypothesis \(H_0\), the likelihood ratio test statistic \(D(\textbf{E})\), as defined in \eqref{eq:lrt_stat}
converges to the $\chi^2$ distribution asymptotically in distribution as the number of random samples goes to infinity, that is:
\begin{equation}
    D(\textbf{E}) \overset{\text{dist}}{\to} \chi^2(d),
\end{equation}
where the degree of freedom $d = 1$. 
\end{lemma}
\edit{Note that the degree of freedom $d$ is the difference in the number of parameters between set $\Lambda_1$ and $\Lambda_0$.}
\begin{lemma}[Neyman-Pearson Lemma~{{\cite[Thm.~8.6-1]{hogg1977probability}}}]\label{lem:neyman}
Suppose there exists a positive constant $c$ and a subset $\mathcal{C}$ of the sample space such that:
\begin{enumerate}
    \item \(\text{Pr}[\textbf{E}\in \mathcal{C}; \Lambda_0] = \tau\)
    \item \(\frac{L(\textbf{E}|\Lambda_0)}{L(\textbf{E}|\Lambda_1)} \leq c~\text{ if } \textbf{E}\in \mathcal{C}\) and 
    \item \(\frac{L(\textbf{E}|\Lambda_0)}{L(\textbf{E}|\Lambda_1)} > c~\text{ if } \textbf{E}\not\in \mathcal{C}\),
\end{enumerate}
then $\mathcal{C}$ is the best critical region of size $\tau$ for testing the simple null hypothesis $H_0:\Lambda = \Lambda_0$ against the simple alternative hypothesis $H_1:\Lambda = \Lambda_1$.
\end{lemma}
In the context of our problem, the best critical region is a detection region of $D(\mathbf{E})$ that maximizes the probability of rejecting the null hypothesis when the alternative hypothesis is true for a given threshold value~(significan\edit{ce} level) $\tau$.
With the key previous results stated, we are prepared to prove the following theorem \edit{about} the properties of Algorithm~\ref{alg:rec_lrt_cp_dtect}.
\begin{theorem}\label{thm:optimal_hypo_test}
    Under Assumption~\ref{asm:exp_dist_res}, the hypothesis test $p \leq \tau$, has the following properties:
    \begin{enumerate}
        \item The test statistic:
        $$D(\mathbf{E}) = 2\left(n\ln(\lambda_n) - (n+1)\ln(\lambda) + 1\right)$$ 
        asymptotically converges to the chi-squared distribution $\chi^2(d)$ in distribution with $d=1$.
        \item The likelihood ratio test 
        $p \leq \tau$
        maximizes the true positive rate (TPR) across
        all hypothesis tests of the same significan\edit{ce} level $\tau$.
    \end{enumerate}
\end{theorem}
\begin{proof}
    Upon a new change point candidate $k_2$, where $E_{k_2}> 0$, 
    let $n = |\mathcal{S}^+_{k_1:k_2-1}|$. The sample mean of \(\mathbf{E} | \mathcal{S}^+_{k_1:k_2-1}\) is computed as:
    \begin{equation}
        \lambda_n = \frac{n}{\sum_{k \in \mathcal{S}^+_{k_1:k_2-1}} E_k},
    \end{equation}
    where $\lambda_{n}$ and $\lambda_{n+1}$ can be expressed in a recursive relationship,
    $
        \lambda_{n+1} = \frac{n+1}{\frac{n}{\lambda_{n}} + E_{k_2}},
    $
    that aligns with the recursive update rule on line~6 \edit{of} Algorithm~\ref{alg:rec_lrt_cp_dtect}.
    The log likelihood $\ln L(\mathbf{E}|\lambda_1)$ is computed through the definition of the probability density function of an exponential distribution:
    \begin{align*}
        \ln L(\mathbf{E}|\Lambda_1) &= \sum_{k \in \mathcal{S}^+_{k_1:k_2-1}} \ln(\lambda_n\exp(-\lambda_n E_k)) \\
        &= \sum_{k \in \mathcal{S}^+_{k_1:k_2-1}} (\ln \lambda_n - \lambda_n E_k)\\
        &= n\ln(\lambda_n) - \lambda_n \sum_{k \in \mathcal{S}^+_{k_1:k_2-1}}E_k\\
        &= n\ln\lambda_n -n.
    \end{align*}
    The likelihood of $H_0$, $L(\mathbf{E}|\Lambda_0)$, is computed similarly \edit{by} replacing 
    $n$ \edit{with} $n+1$. 
    Let $\lambda \coloneqq \lambda_{n+1}$ to simplify notation.
    Then, the test statistic $D(\mathbf{E})$ can be evaluated as:
    \begin{align*}
        D(\mathbf{E}) &= 2\left(n\ln(\lambda_n) - (n+1)\ln(\lambda) + 1\right)\\
        &= -2((n+1)\ln\lambda -(n+1) - (n\ln\lambda_n -n))\\
        &= -2(\ln(L(\mathbf{E}|\Lambda_0)) -\ln(L(\mathbf{E}|\Lambda_1)))\\
        &= -2\left(\ln\frac{L(\mathbf{E}|\Lambda_0)}{L(\mathbf{E}|\Lambda_1)}\right),
    \end{align*}
    which is in the form of likelihood ratio test statistic in Definition~\ref{def:lrt_statistic}.
    By Lemma~\ref{lem:wilks}, the test statistic $D(\mathbf{E}) \overset{\text{dist}}{\to} \chi^2(d)$, where
    the degree of freedom $d = d_{\text{alt}} - d_{\text{null}} = 1$ by Definition~\ref{def:lrt_statistic}.
    Moreover, the likelihood ratio test:
    \begin{equation}\label{eq:lrt_ineq}
        1 - \text{CDF}_{\chi^2}(D(\mathbf{E})) \leq \tau
    \end{equation}
    induces a critical region $\mathcal{C} = \{\mathbf{E} \in \Omega : \text{reject } H_0 \text{ if } \mathbf{E}\in \mathcal{C}\}$ of size $\tau$, \edit{where $\Omega$ is the sample space of $\mathbf{E}$}.
    The sample $E_k \in \mathcal{C}$ if and only if \eqref{eq:lrt_ineq} is satisfied.
    Since $\text{CDF}_{\chi^2}(\cdot)$ and $\ln(\cdot)$ are monotonic functions, we can rewrite \eqref{eq:lrt_ineq}:
    \begin{align*}
        &~1 - \text{CDF}_{\chi^2}(D(\mathbf{E})) \leq \tau\\
        \Rightarrow& D(\mathbf{E}) \geq \text{CDF}_{\chi^2}^{-1}(1 - \tau)\\
        \Rightarrow& \ln\left(\frac{L(\mathbf{E}|\Lambda_0)}{L(\mathbf{E}|\Lambda_1)}\right) \leq \frac{1}{2}\text{CDF}_{\chi^2}^{-1}(1 - \tau)\\
        \Rightarrow& \frac{L(\mathbf{E}|\Lambda_0)}{L(\mathbf{E}|\Lambda_1)} \leq \exp\left(\frac{1}{2}\text{CDF}_{\chi^2}^{-1}(1 - \tau)\right) \eqqcolon c.
    \end{align*}
    The critical region $\mathcal{C}$ is the best among all hypothesis tests through Lemma~\ref{lem:neyman}.
\end{proof}
Algorithm~\ref{alg:rec_lrt_cp_dtect} maximizes the true positive rate 
, a.k.a. power, at a given significance level $\tau$. 
However, we have chosen the simple EWMA model and exponentially distributed error distribution to simplify the computation of the test statistic $D(\mathbf{E})$ in the recursive context of our problem. 
Therefore, the fulfillment of Assumption~\ref{asm:exp_dist_res} may vary depending on the complexity of the online parameter estimator and the level of excitation of the upstream regressor signal. 

The change point detection algorithm can be applied to all online estimation algorithms that retain some memory or state during their estimation process. 
By integrating the downstream change point algorithm, we can enhance their performance in a time-varying environment. 

\begin{algorithm}
    \caption{Change Point State Resetting Online Estimator}\label{alg:cps_reset_scheme}
    \begin{algorithmic}[1]
        \Initialize{
            Estimator Hyper-parameters: $\Psi$\\
            Online Estimator State: $\mathcal{S}_0$\\
            Parameter Estimate: $\hat\theta_0$\\
            Online Estimator: $\hat{\theta} \gets \hat{\theta}(\ \cdot\ ,\mathcal{S}_0, \hat\theta_0; \Psi)$
        }
        \Inputs{
            Input datum: $(x_{k}, y_{k})$, Data count: $k$
        }
        \If{$k \geq 1$}
        \State $\hat{\theta}_{k}, \mathcal{S}_k \gets \hat{\theta}((x_{k}, y_{k}), \mathcal{S}_{k-1}, \hat{\theta}_{k-1})$
        \State $Y_k \gets -\log(\|y_k - \phi(x_k)\hat{\theta}_{k-1}\|_2^2)$
        \State \textbf{Run} Algorithm~\ref{alg:rec_lrt_cp_dtect}: detectChange $\gets$ rLRT($Y_k$)
        \If{detectChange}
            \State $\hat{\theta} \gets \hat{\theta}(\ \cdot\ , \mathcal{S}_0, \hat{\theta}_k)$
        \Else
            \State $\hat{\theta} \gets \hat{\theta}(\ \cdot\ , \mathcal{S}_k, \hat{\theta}_k)$
        \EndIf
        \EndIf
    \end{algorithmic}
\end{algorithm}
Presented in Algorithm~\ref{alg:cps_reset_scheme} is an algorithmic framework for incorporating the change point detection algorithm to facilitate online state resetting.
For a exponentially forgetting recursive least squares algorithm, the state is the approximated covariance matrix, $\mathcal{S} = \{P_k\}$.
In the case of GW-RLS, the states include the information of the excitation set $H^{(e)}_k$, the excitation set cross-correlation vector $\nu_k^{(e)}$, the excitation set $\Phi^{(e)}_k$, and the approximated covariance matrix $P_k$.

We may also consider the previously estimated parameter $\hat\theta_k$ as part of the estimator states 
if the system parameters deviate too much from their previous values after abrupt changes.

\section{Simulations}\label{sec:sim}
\edit{This section is organized into three subsections: 
In Section~\ref{sec:sim:simpleSIS}, we illustrate the performance of the proposed GW-RLS on a simple SIS system with two parameters, similar to Example~\ref{ex:lack_pe_networked_SIS_SIR}.
In Section~\ref{sec:sim:networkedSIR}, we test the performance of Algorithm~\ref{alg:g_rls} on a networked SIR system with $56$ parameters. 
In Section~\ref{sec:sim:timevarying}, the joint performance of Algorithm~\ref{alg:g_rls} and Algorithm~\ref{alg:rec_lrt_cp_dtect} coupled under Algorithm~\ref{alg:cps_reset_scheme} is illustrated using the two-node SIR system with three change points from Example~\ref{ex:sir_cp}.
These numerical experiments aim to illustrate the performance of the proposed algorithms in multiple scenarios.
}

\subsection{SIS Compartmental Model}\label{sec:sim:simpleSIS}
\edit{This section compares the performance of GW-RLS with the classical EF-RLS~\cite{islam2019recursive} and a gradient descent-based algorithm: initially excited multi-model adaptive identifier (IE-MMAI)~\cite{dhar2022initial} with insufficient persistent of excitation in the regressor signal.
Detailed discussions are presented below to offer insights into the disparities in performance, focusing on the practical identifiability and persistent excitation perspectives.
}

Consider the \edit{discretized} SIS compartmental dynamics with process noise:
\edit{
\begin{align}\label{eq:single_node_SIS}
    X_{k+1} &= X_{k} + h\phi(X_k) + \xi_k \nonumber\\
    \xi_k &= \sigma B_k, \qquad B_k \sim \mathcal{N}(0, h),
\end{align}}
where $X_k$ is the proportion of infected population at time $t_k$, 
the parameter $\theta = [\beta,\ \gamma]^\top$, 
such that $\beta$ is the infection rate, $\gamma$ is the recovery rate, 
the regressor $\phi(X_k) = [X_k(1 - X_k),\ -X_k]$.
%
\edit{Note that \eqref{eq:single_node_SIS} reduces to a 
Euler discretization of a deterministic ODE
when the standard deviation $\sigma$ or step size $h$ approaches zero.
In particular, we choose $h=0.1$ for the simulation of this subsection.
}

We present parameter estimation results for both the noise-free case and including Gaussian process and observation noise for the SIS dynamics in~\eqref{eq:single_node_SIS}.

\begin{figure}
    \centering
    \begin{subfigure}[t]{0.49\columnwidth}
            \includegraphics[width=\columnwidth, trim = 0 0cm 0 .1cm, clip]{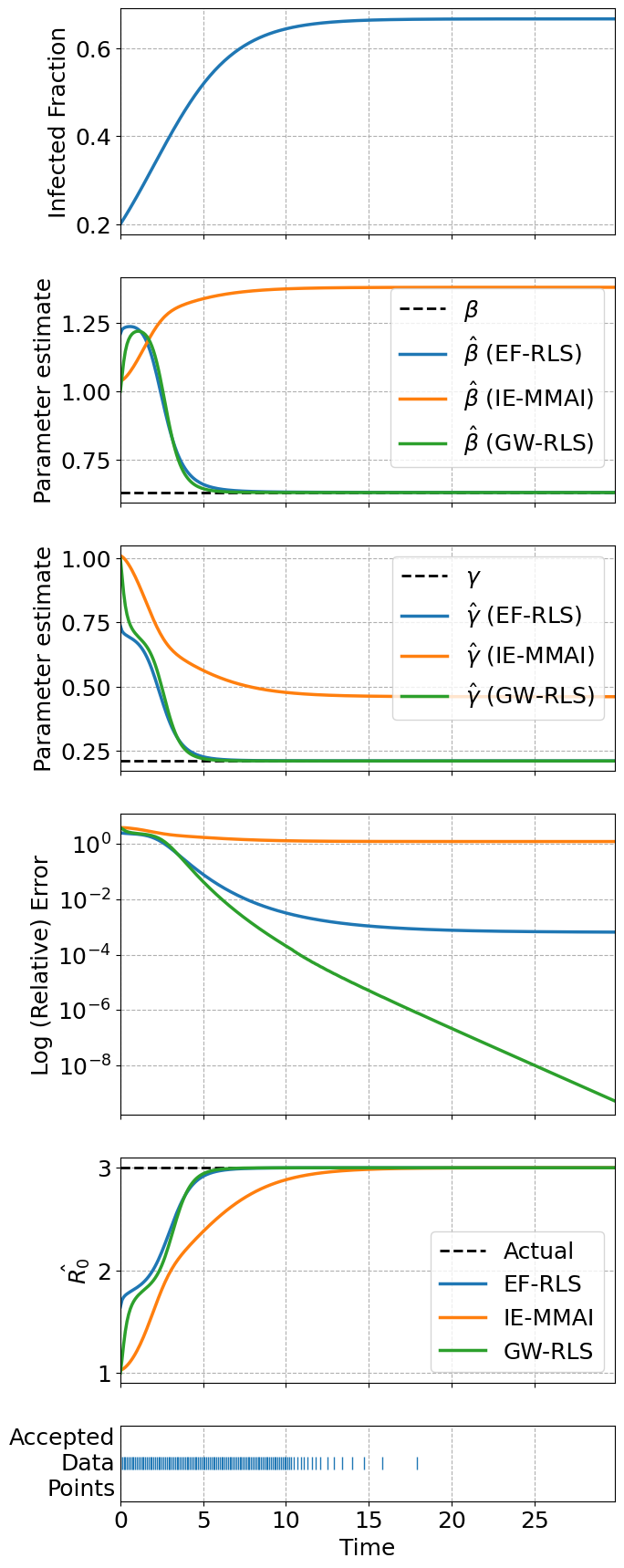}
    \end{subfigure}
    \begin{subfigure}[t]{0.49\columnwidth}
            \includegraphics[width=\columnwidth, trim = 0 0cm 0 .1cm, clip]{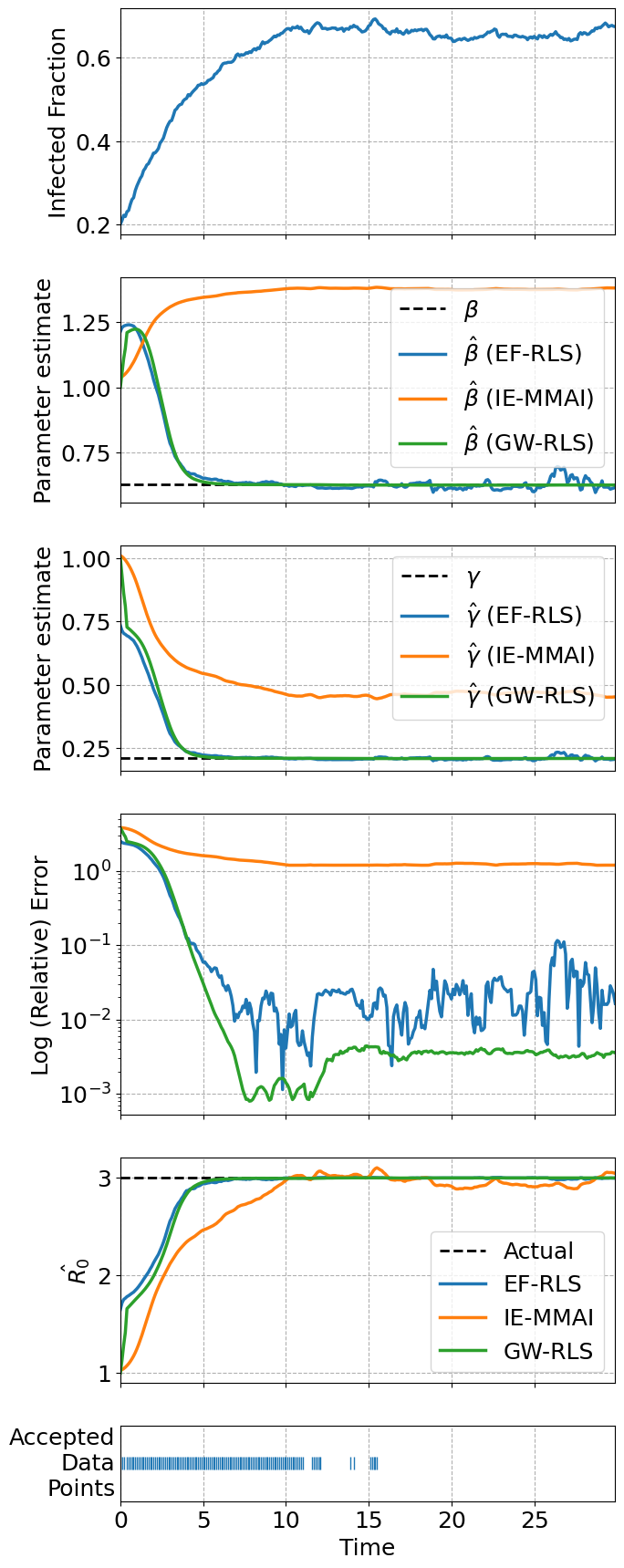}
    \end{subfigure}
    \caption{
        The top panel shows the infected proportion under the SIS model (\(\beta = 0.8076\), \(\gamma = 0.2692\)).
        The bottom panel depicts data indices accepted into the greedy excitation set by \edit{GW-RLS}.
        Then, the rest of the panels, 
        arranged from top to bottom, compare the performance of EF-RLS, IE-MMAI, and \edit{GW-RLS}: 
        (\edit{Second to} Top)~parameter estimates over time for both the noise-free case~(left panel) and a noisy simulation~(right panel); 
        (Middle)~maximum relative error in log scale for parameter estimates over time; 
        (\edit{Second to} Bottom)~reproduction number estimates over time.
    }
    \label{fig:compare_scalar_SIS}
    \vspace{-.3cm}
\end{figure}
In Fig.~\ref{fig:compare_scalar_SIS}, we compare the performance of  EF-RLS and IE-MMAI~\cite{dhar2022initial}, to the performance of the \edit{GW-RLS} Algorithm we propose. 
The same initial estimates of parameters, $\theta_0 = [\beta_0, \gamma_0]^\top = [1, 1]^\top$,
are used for all algorithms with the exception of IE-MMAI, for which the \(m=3\) models were initialized randomly around 
\(\theta_0\).
The forgetting rate of EF-RLS and \edit{GW-RLS} are set to $0.98$.
IE-MMAI, designed for only LTI systems and as a gradient descent-like, first-order method, is justifiably sensitive to the choice of the initial parameter estimates, and often fails to converge due to the poor practical identifiability of SIS models, as discussed in Example~\ref{ex:lack_pe_networked_SIS_SIR}.
Note that the estimated reproduction numbers still converge to the actual value despite the lack of convergence of the parameter estimates themselves, which is consistent with our discussion of SIS practical identifiability in Section \ref{sec:lack-pe-pi}.
On the other hand, even in the presence of noise, the parameter estimates converge for both EF-RLS and \edit{GW-RLS}. 

Data accepted into the \edit{GW-RLS} greedy excitation set (\ref{def:greedyset}) used to construct the main regressor are diagrammatically depicted \edit{are during} the lowermost block of Fig.~\ref{fig:compare_scalar_SIS}. 
A majority of points accepted by the algorithm is in the transient rise of states before the equilibrium is reached. 
The beginning of the epidemic garners a critical amount of information about the epidemic parameters as would be expected.

Note that while EF-RLS is comparable in performance to \edit{GW-RLS} for the noise-free case (left panel of Fig.~\ref{fig:compare_scalar_SIS}), it becomes increasingly oscillatory upon losing excitation in the noisy case (right panel of Fig.~\ref{fig:compare_scalar_SIS}). 
A closer look at the covariance matrix \(P_k\) for both algorithms reveals a steady increase in the condition number and maximum eigenvalue of \(P_k\) for EF-RLS, while those of \edit{GW-RLS} saturate due to its tendency to avoid picking up non-exciting data points (see Fig.~\ref{fig:Pk_RLS}). 
This phenomenon in EF-RLS is observed in both the noise-free and noisy cases and is known in the literature as covariance windup \cite{fortescue1981selftune}, which occurs when a non-unity forgetting factor in EF-RLS causes \(P_k\) to get closer to a singular matrix upon losing persistence of excitation. 
The linear increase in the maximum eigenvalue of \(P_k\) is consistent with past analysis~\cite{cao2001windup}. 
Upon running the parameter estimation task for longer times, the RLS estimates diverge. 

\begin{figure}
    \centering
    \includegraphics[width=0.9\columnwidth]{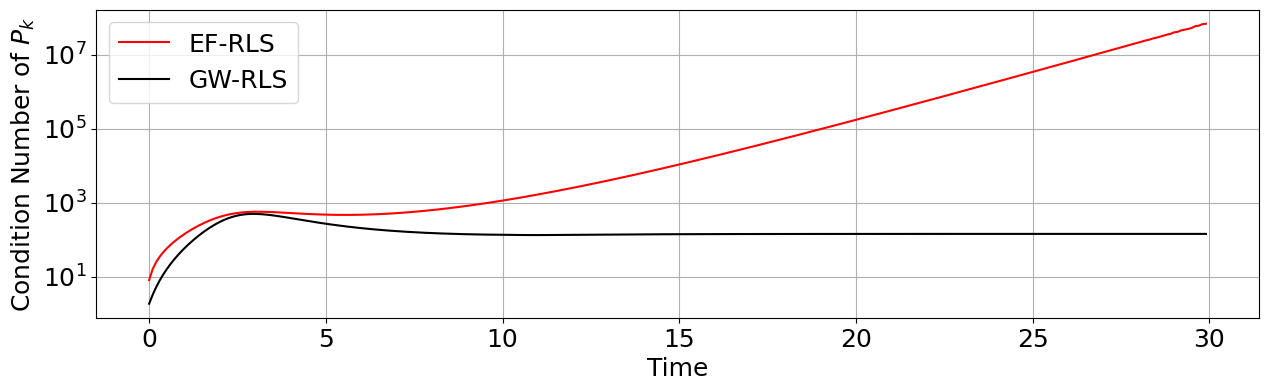}
    \caption{\footnotesize{Plot of the condition number of the covariance matrix \(P_k\) against time for both EF-RLS and \edit{GW-RLS} running on data from the noise-free SIS simulation (Fig.~\ref{fig:compare_scalar_SIS}).}}
    \label{fig:Pk_RLS}
\end{figure}
\subsection{Networked SIR Model}\label{sec:sim:networkedSIR}
In this section, we aim to investigate the performance of the \edit{GW-RLS} algorithm in estimating epidemic parameters under more realistic conditions involving networked structures \edit{with step size $h=0.2$}.
Three network topologies are considered in the simulations, which include a fully-connected, a star, and an Erd\H{o}s-R\'{e}nyi (ER)~\cite{erdds1959random} network.
Fig.~\ref{fig:3graphs} shows the three topologies of networked structures we use in this section.
\begin{figure}
    \centering
    \begin{subfigure}[h]{.32\columnwidth}
        \centering
        \includegraphics[width=\columnwidth]{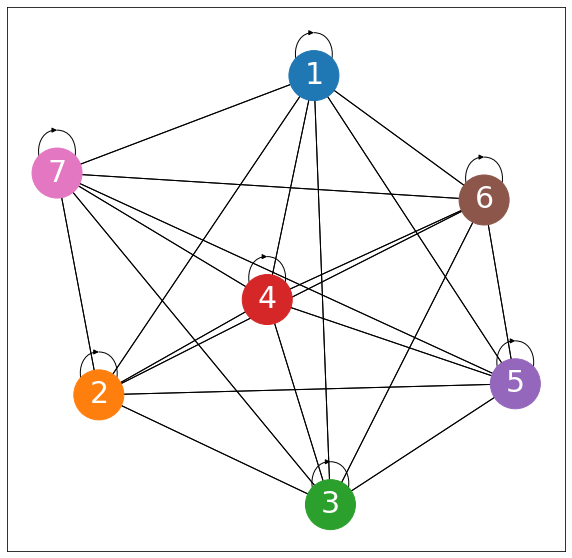}
        \caption{Fully-connected}
        \label{fig:graph:fc}
    \end{subfigure}
    \begin{subfigure}[h]{.32\columnwidth}
        \centering
        \includegraphics[width=\columnwidth]{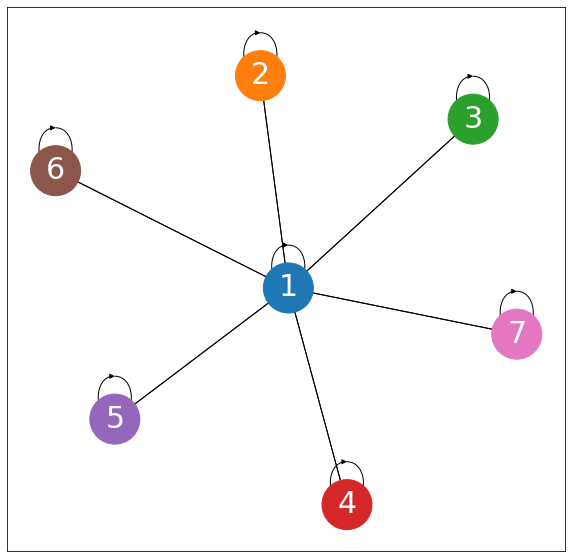}
        \caption{Star}
        \label{fig:graph:star}
    \end{subfigure}
    \begin{subfigure}[h]{.32\columnwidth}
        \centering
        \includegraphics[width=\columnwidth]{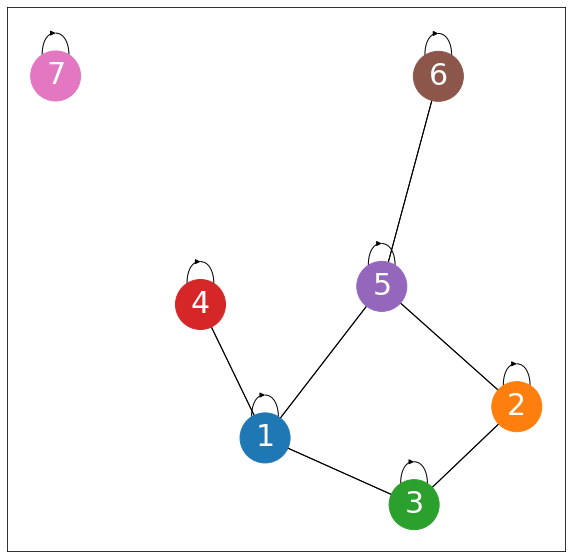}
        \caption{Erd\H{o}s-R\'{e}nyi}
        \label{fig:graph:er}
    \end{subfigure}
    \caption{Fig.~\ref{fig:graph:fc} to \ref{fig:graph:er} illustrate the structure of the 7-node networks, where the Erd\H{o}s-R\'{e}nyi (ER) network in Fig.~\ref{fig:graph:er} is generated with \(0.50\) edge inclusion probability.}
    \label{fig:3graphs}
\end{figure}

To generalize the $2$-node SIR network, 
We can rewrite the $2$-node SIR network presented in~\eqref{eq:sir_cp}.
Let $n$ be the number of nodes in a network, 
we choose the regressor \(\phi_k \in \mathbb{R}^{(2n) \times (n^2+n)}\) as follows:
\begin{equation}\label{eq:SIRnetworkphi}
\begin{small}
    \phi_k :=
    \begin{bmatrix}
        I(t_k)^\top \otimes \text{diag}(1_n - I(t_k) - R(t_k)) & - \text{diag}(I(t_k))\\
        0_{n\times n^2} & \text{diag}(I(t_k))
    \end{bmatrix},
\end{small}
\end{equation}
and the target vector as $\psi_k = \begin{bmatrix}
    \dot I(t_k)^\top & \dot R(t_k)^\top
\end{bmatrix}$,
where 
\(I(t) \in [0, 1]^n\) and \(R(t) \in [0, 1]^n\) represent the infected and recovered proportion of the populations accordingly.
The parameter vector \(\theta = [\beta, \gamma]^\top\) is the concatenated flattened adjacency matrix \(\beta = \text{vec}(B)\) (of infection rates between different nodes, not assumed symmetric) and recovery rates vector \(\gamma\), such that $\theta \in \mathbb{R}^{n^2 + n}_{\geq 0}$ is a (\(n^2 + n\))-dimensional vector.
Fig.~\ref{fig:fc_sir_noiseless} shows the infection levels and their derivatives in the noiseless $7$-node network. 
\begin{figure}
    \centering
    \includegraphics[width=\columnwidth]{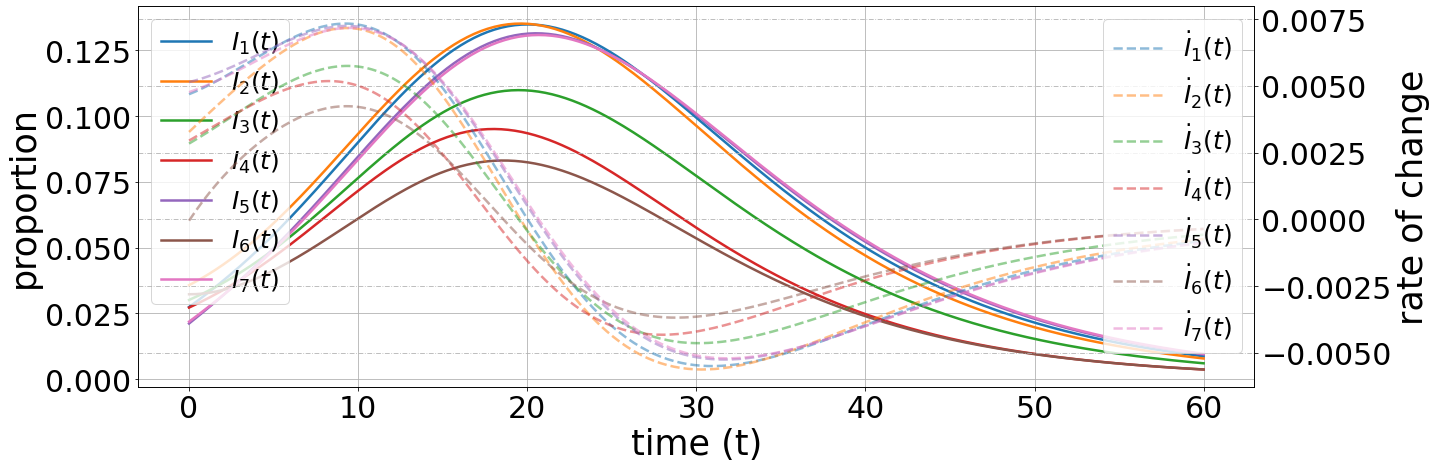}
    \caption{State trajectories of fully-connected networked SIR model without noise.}
    \label{fig:fc_sir_noiseless}
\end{figure}

Before we begin discussing the performance of the system parameters' estimation, it is helpful to introduce two more notions to help summarize the estimation performance over $50$ parameters.
In Example~\ref{ex:lack_pe_networked_SIS_SIR}, we introduced the basic reproduction number $\mathcal{R}_0$ of a single-node SIR/SIS system as the ratio between the infection rate and the recovering rate $\beta/\gamma$. 
The local basic reproduction number generalizes the notion of basic reproduction number to networked epidemics.
\begin{definition}[Local Basic Reproduction Number]\label{def:local_basic_r_0}
    The local basic reproduction number $\mathcal{R}^0_i \in \mathbb{R}_{\geq 0}$~\cite{she2023distributed} of node $i$ in a networked SIR system with parameterization $\theta_{\text{SIR}} = (B, \gamma)$ is:
    \begin{equation}
        \mathcal{R}^0_i \coloneqq \sum_{j=1}^{n} [\text{diag}(\gamma)^{-1} B]_{ij}.
    \end{equation}
\end{definition}
\begin{definition}[Relative Error]
    The relative error is defined as $\tilde{\theta}(t) \coloneqq |\frac{\theta_{\text{true}} - \hat{\theta}(t)}{\theta_{\text{true}} + \delta}|$. 
\end{definition}
We choose $\delta = 10^{-2}$ to be the order of magnitude of the smallest positive parameter in the network.
\begin{definition}[Relative Error Profile]
    Let $p$ be the number of parameters in $\theta$.
    The relative error profile is defined as the area enclosed by the interval $[\min_{i \in [p]}\{\tilde{\theta}_i(t)\}, \max_{i \in [p]}\{\tilde{\theta}_i(t)\}]$ for all $t_0\leq t \leq T$, where $\tilde{\theta} \in \mathbb{R}^{p}$ is a vector with length $p > 1$.
\end{definition}
Since we have compared the performance of EF-RLS and \edit{GW-RLS} with the single-node SIS example, we now proceed to compare the performance of \edit{GW-RLS} and the Directional Forgetting Recursive Least Squares (DF-RLS) algorithm~\cite{cao2000directional}.
These two algorithms share similar working principles, therefore, the comparison is carried out using the noiseless fully-connected SIR network example shown in Fig.~\ref{fig:fc_sir_noiseless}.
\begin{figure}
    \centering
    \includegraphics[width=\columnwidth]{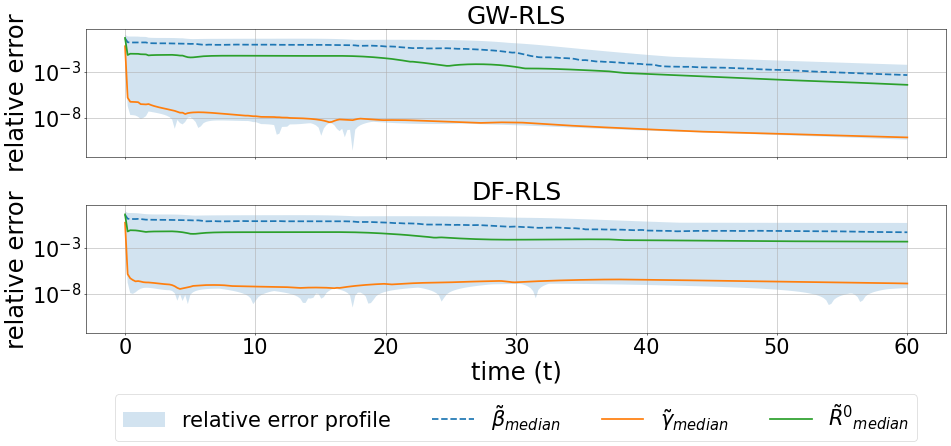}
    \caption{Performance comparison between \edit{GW-RLS} and DF-RLS on the noiseless $7$-node fully-connected SIR network.}
    \label{fig:GWRLS_DFRLS_noiseless}
\end{figure}
We show, in Fig.~\ref{fig:GWRLS_DFRLS_noiseless}, that \edit{GW-RLS} performs better than DF-RLS in terms of the \edit{final} relative error in the estimation median of $\beta, \gamma,$ and $\mathcal{R}_0$ \edit{at time $t=60$}.
\edit{Table~\ref{tab:comp_GWRLS_DFRLS} compares the final relative error of GW-RLS and DF-RLS.}

Furthermore, when compared to the single-node SIS parameter estimation in Fig.~\ref{fig:compare_scalar_SIS}, the relative estimation error increases significantly in GW-RLS, by orders of magnitude.
The increase in estimation error is partially due to the curse of dimensionality in the estimation problem.
For every node added to the network, there could be $n-1$ extra edges, i.e., 
the number of system parameters grows quadratically, $\mathcal{O}(n^2)$, with respect to the number of measured signals,~$n$.

In the next example, we allow the process noise to scale with the square root of the state vector $X(t) \coloneqq \begin{bmatrix}
    I^\top(t) & R^\top(t)
\end{bmatrix}^\top$, approximating the arrival of random infection and recovery events under a Poisson distribution\cite[pg.~197]{keeling2008modeling}.
\edit{Consider the \edit{discretized} networked SIR compartmental dynamics with process noise:}
\edit{
\begin{align}\label{label:sir_network_sde}
    X_{k+1} &= X_{k} + h\phi(X_k) + \xi_k \nonumber\\
    \xi_k &= \sqrt{\phi(X_k)\text{diag}(\theta)} B_k, \qquad B_k \sim \mathcal{N}(0, h),
\end{align}
with $\phi(X_k) \coloneqq \phi_k$ as defined in \eqref{eq:SIRnetworkphi}, and $h> 0$ denoting the size of the discrete time step,
}
\edit{
under the same assumption that $\psi_k = \lim_{h\to 0} {(X_{k+1} - X_k)}/{h}$
is known.}
The synthetic data generated by the Euler–Maruyama method~\cite{maruyama1955continuous} is shown in Fig~\ref{fig:NetworkSIR}.
\begin{figure}
    \centering
    \includegraphics[width=\columnwidth]{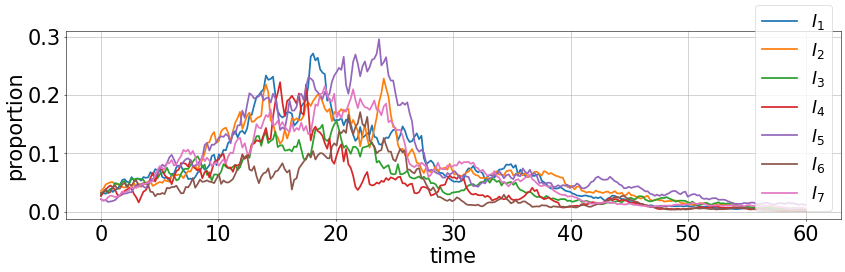}
    \caption{Simulated infection pervalence on the fully-connected network with process noise.}
    \label{fig:NetworkSIR}
\end{figure}

We observe, in Fig.~\ref{fig:NetworkSIR_est_performance}, that the estimation performance of both the \edit{GW-RLS} and DF-RLS algorithms improves significantly with process noise. This improvement is caused by the unconstrained trajectory of $\phi_k^\top\phi_k$ in the simple SIR dynamics. 
The phenomenon aligns with the intuition we learn from \edit{Proposition}~\ref{thm:general_lack_pe} and Example~\ref{ex:lack_pe_networked_SIS_SIR}.
In practice, the above experiments indicate that process noise might improve identification results with proper techniques \edit{for} recovering state derivatives~\cite{chartrand2011numerical}.

\begin{figure}
    \centering
    \includegraphics[width=\columnwidth]{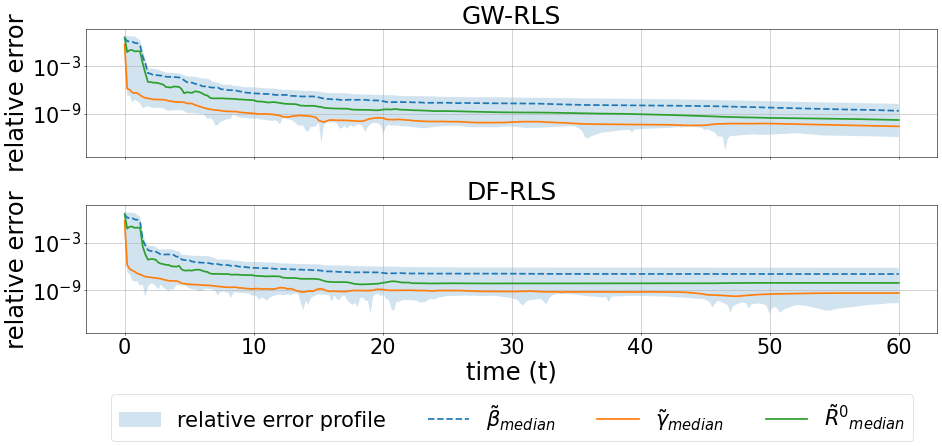}
    \caption{
    Performance comparison between \edit{GW-RLS} and DF-RLS on the $7$-node fully-connected SIR network with known process noise. The estimation with process noise outperforms that of the noiseless example.
    } \label{fig:NetworkSIR_est_performance}
\end{figure}
\begin{table}
    \centering
    \renewcommand{\tablename}{\centering \textbf{Table}}
    \resizebox{.49\textwidth}{!}{%
        \begin{tabular}{lcccccc}
            \hline
            & \textbf{$\beta_{\text{median}}$} & \textbf{$\gamma_{\text{median}}$} & \textbf{$\mathcal{R}^0_{\text{median}}$} \\
            \hline
            \multicolumn{4}{c}{\textbf{Noiseless}} \\
            \hline
            GW-RLS & 0.000856 & $7.68 \times 10^{-11}$ & $3.88 \times 10^{-5}$ \\
            DF-RLS & 0.051253 & $1.39 \times 10^{-7}$ & $4.93 \times 10^{-3}$ \\
            \hline
            \multicolumn{4}{c}{\textbf{Noisy}} \\
            \hline
            GW-RLS & $1.23 \times 10^{-9}$ & $6.63 \times 10^{-12}$ & $1.71 \times 10^{-10}$ \\
            DF-RLS & $1.09 \times 10^{-7}$ & $4.40 \times 10^{-10}$ & $8.36 \times 10^{-9}$ \\
            \hline
        \end{tabular}
    }
    \caption{Final relative error of GW-RLS and DF-RLS in Noiseless and Noisy Conditions}
    \label{tab:comp_GWRLS_DFRLS}
\end{table}
To conclude, we examine the performance of 
\edit{GW-RLS on estimating a system with 56 parameters}
under a noiseless scenario. 
Fig.~\ref{fig:sparcity_error} illustrates the true values and estimation errors of the infection network previously introduced in Fig.~\ref{fig:3graphs}. 
We observe that the relative error is highest for the estimation of empty edges, indicating the utility of LASSO regression in contrast to the ridge regression currently used in \edit{GW-RLS}.
However, as demonstrated in the comparison between Fig.~\ref{fig:GWRLS_DFRLS_noiseless} and \ref{fig:NetworkSIR_est_performance}, the presence of persistently exciting input signals is also crucial for identifying high-dimensional sparse networks. 
\begin{figure}
    \centering
    \begin{subfigure}[t]{0.48\columnwidth}
            \includegraphics[width=\columnwidth]{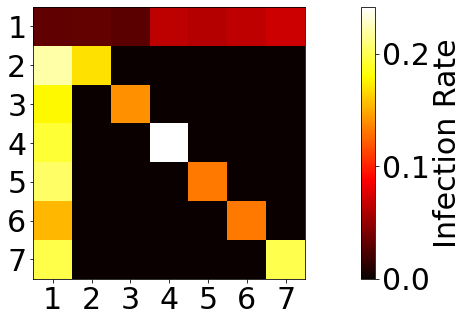}
            \caption{True network (Star)}
            \label{fig:mat:star_pattern}
            \vspace*{2mm}
    \end{subfigure}
    \begin{subfigure}[t]{0.48\columnwidth}
            \includegraphics[width=\columnwidth]{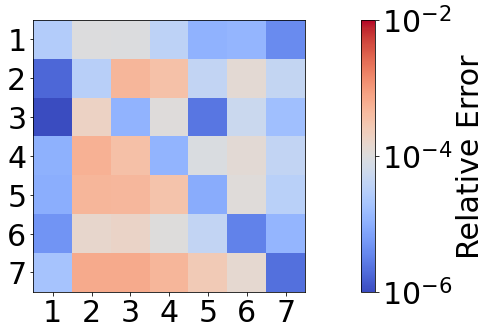}
            \caption{Estimation errors (Star)}
            \label{fig:mat:star_error}
            \vspace*{2mm}
    \end{subfigure}

    \begin{subfigure}[t]{0.48\columnwidth}
            \includegraphics[width=\columnwidth]{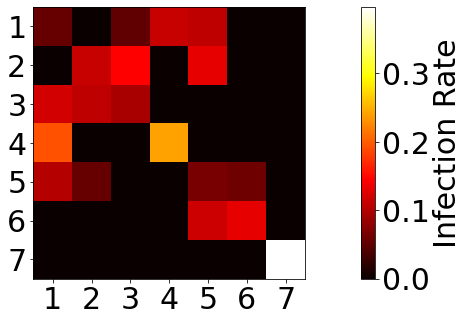}
            \caption{True network (ER)}
            \label{fig:mat:er_pattern}
            \vspace*{2mm}
    \end{subfigure}
    \begin{subfigure}[t]{0.48\columnwidth}
            \includegraphics[width=\columnwidth]{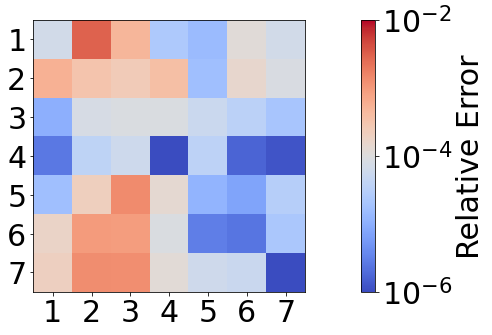}
            \caption{Estimation errors (ER)}
            \label{fig:mat:er_error}
            \vspace*{2mm}
    \end{subfigure}
    \caption{        Fig.~\ref{fig:mat:star_pattern} and Fig.~\ref{fig:mat:er_pattern} illustrate the matrix form of the star and ER network shown earlier in Fig.~\ref{fig:graph:star} and Fig.~\ref{fig:graph:er}.
        Fig.~\ref{fig:mat:star_error} and Fig.~\ref{fig:mat:er_error} show the relative estimation error of their corresponding rate matrix at time $t=60$.
    }
    \label{fig:sparcity_error}
\end{figure}

\subsection{Time-Varying System}\label{sec:sim:timevarying}
In this section, we examine the effectiveness of the change point detection algorithm proposed in Section~\ref{sec:fault-dect} by assessing parameter estimation errors. 
We start by considering 
Example~\ref{ex:sir_cp} with process noise that is similar to the experimental setup in the previous section, \edit{with step size $h=0.2$}.
We also added observation noise to the observed signal in the form of \eqref{eqn:gen_nonlinear_sys_pre} with a strength of around one-tenth the value of $\psi_k$.
Fig.~\ref{fig:2-node-time-v-inf-data} 
\edit{shows the infection prevalence of the two-node SIR epidemic model with process noise, as specified in Example~\ref{ex:sir_cp} with noise structure similar to \eqref{label:sir_network_sde}.}

\begin{figure}
    \centering
    \includegraphics[width=\columnwidth]{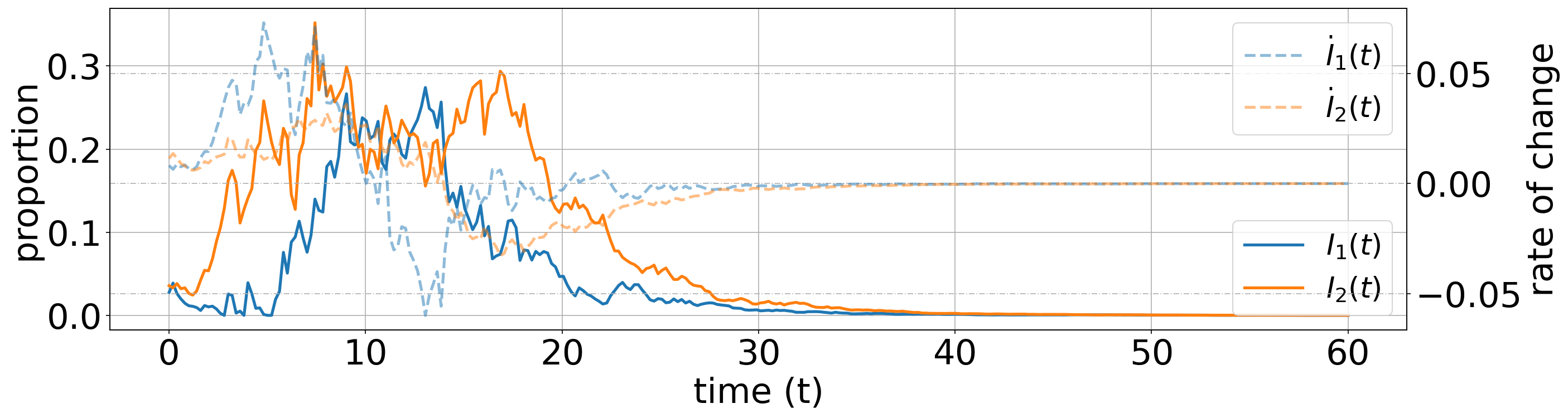}
    \caption{Infection prevalence of a 2-node SIR network}
    \label{fig:2-node-time-v-inf-data}
\end{figure}
\begin{figure}
    \centering
    \begin{subfigure}[t]{\columnwidth}
        \centering
        \includegraphics[width=\columnwidth]{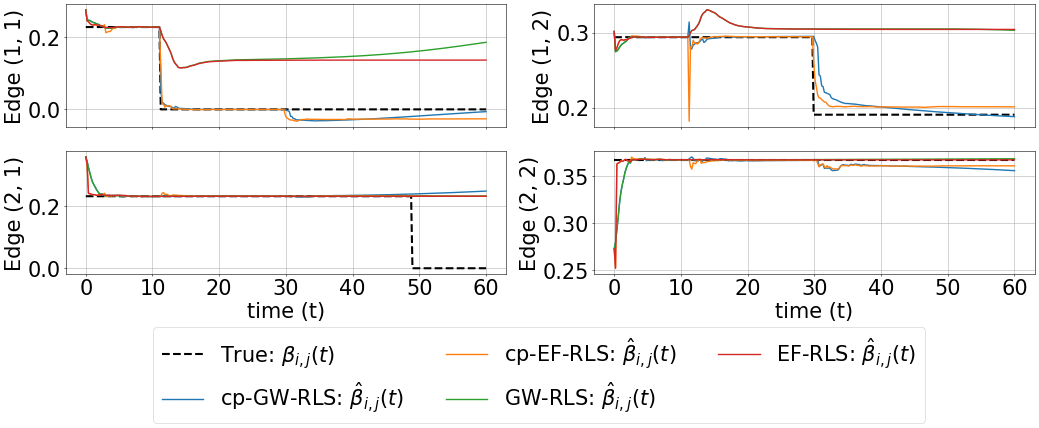}
        \caption{Infection rate estimation}
        \label{fig:time_varying_tracking_performance_beta}
    \end{subfigure}
    \begin{subfigure}[t]{\columnwidth}
        \centering
    \end{subfigure}
    \begin{subfigure}[t]{\columnwidth}
        \centering
        \includegraphics[width=\columnwidth]{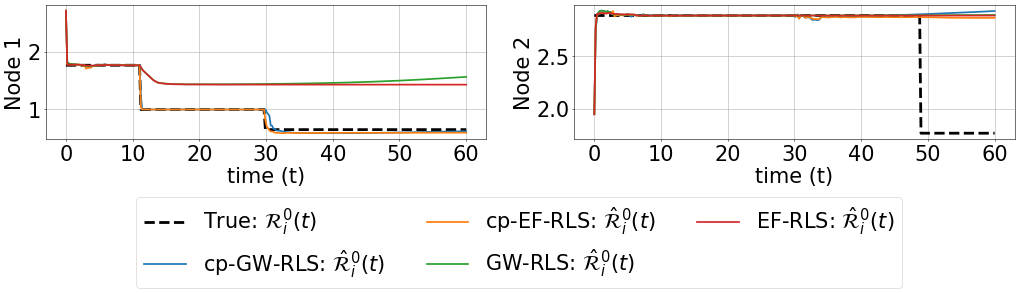}
        \caption{Local basic reproduction number estimation}
        \label{fig:time_varying_tracking_performance_r0}
    \end{subfigure}
    \caption{Parameter tracking performance comparison before and after applying change point detection resetting.
    \edit{The cp-GW-RLS and cp-EF-RLS represent the estimation trajectory under GW-RLS and EF-RLS when coupled with the proposed change point state resetting method described in Algorithm~\ref{alg:rec_lrt_cp_dtect}.}
    }
    \label{fig:time_varying_tracking_performance}
\end{figure}

\edit{
Fig.~\ref{fig:time_varying_tracking_performance} demonstrates that the parameter estimation performance significantly improves when coupled with the change point state resetting scheme. In Fig.~\ref{fig:time_varying_tracking_performance_beta}, we observe that while EF-RLS and GW-RLS respond to the first parameter jump, their adaptive behavior is further enhanced when combined with the proposed change point state resetting algorithm. 
Coupling with Algorithm~\ref{alg:cps_reset_scheme} allows cp-EF-RLS and cp-GW-RLS to successfully adapt to both the first and second parameter changes. 
Notably, the second change point occurs at $t=30$, coinciding with the state signal decay, whereas the final change point at $t=49$ is nearly undetectable.
}
\begin{figure}
    \centering
    \includegraphics[width=\columnwidth]{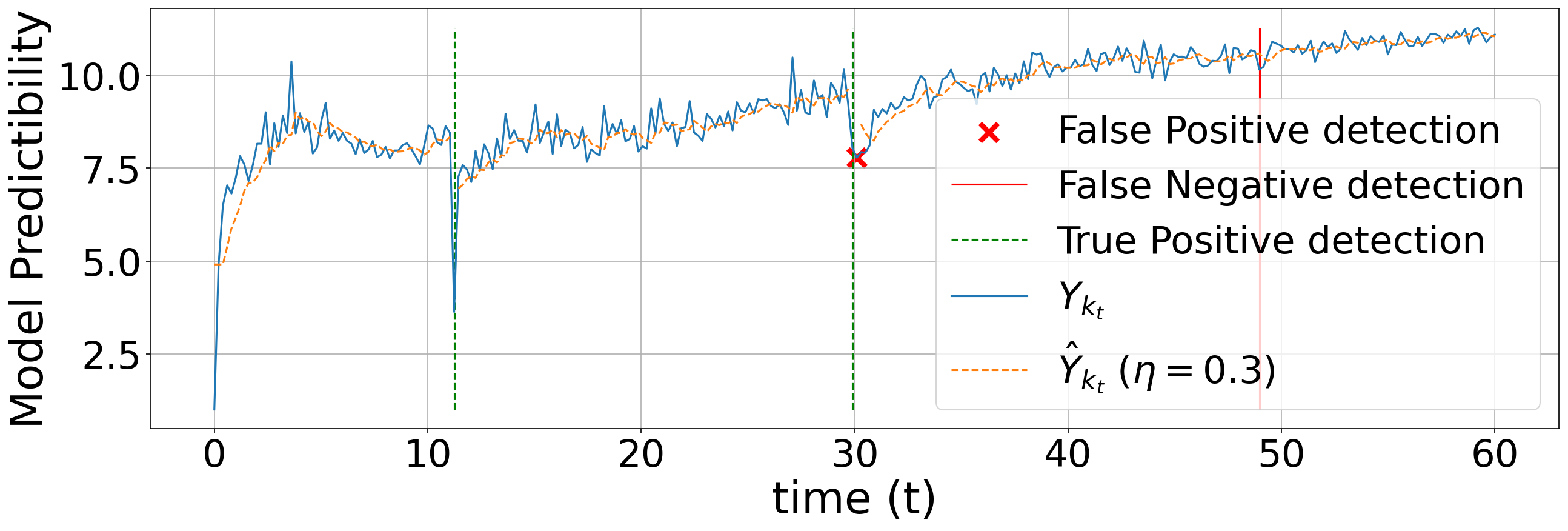}
    \caption{Change point detection performance with GW-RLS, significance level is chosen to be $\tau= 0.1$ with exponential weight $\eta=0.5$.
    The y-axis represents the magnitude of the model predictability $Y_k$, with higher values being preferable.
    }
    \label{fig:chp:rec_lrt_ch_performance}
\end{figure}

Fig.~\ref{fig:chp:rec_lrt_ch_performance} shows the performance of Algorithm~\ref{alg:rec_lrt_cp_dtect} in terms of the true-positive, false-positive, and false-negative detection \edit{instances}. \edit{When the frequency of the parameter changes is low, Algorithm~\ref{alg:rec_lrt_cp_dtect} performs relatively well, and the effect of false positive and negative detections does not negatively impact future predictability, in this example.}

\edit{
Next, we are interested in the change point detection performance under different hyperparameter setups.
In particular, we illustrate the intuition behind choosing different $\eta$ and its impact on the change point detection performance.
}
\begin{figure}
    \centering
    \includegraphics[width=\columnwidth]{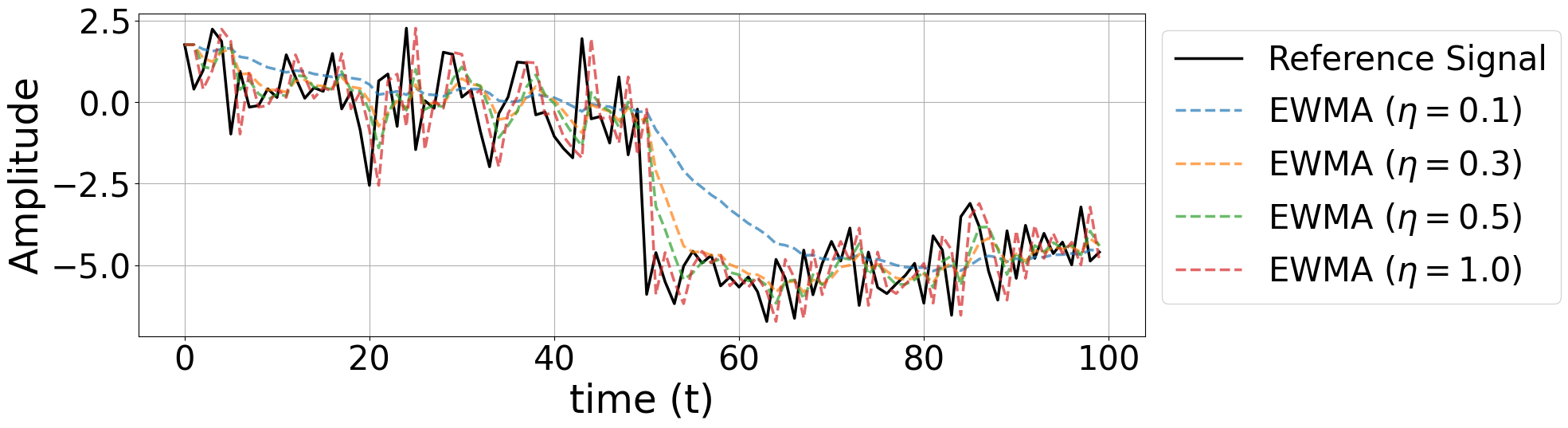}
    \caption{The EWMA acts as a low pass filter on a Gaussian white noise signal with an abrupt downward drift at time $t=50$.
    Values of $\eta \in \{0.1, 0.3, 0.5, 1.0\}$ determine the smoothing effect of the EWMA.
    }
    \label{fig:chpt:EWAM_lagged_effects}
\end{figure}
In Fig.~\ref{fig:chpt:EWAM_lagged_effects}, we \edit{observe} that the exponential weight $\eta \in [0, 1]$ determines the smoothing intensity of an EWMA.
When $\eta = 1$, EWMA corresponds to the most recent update and provides no smoothing effect.
As $\eta$ approaches $0$, EWMA relies solely on its initial value, resulting in maximal smoothing. 

\edit{Next, we increase the number of change points and check the performance of Algorithm~\ref{alg:rec_lrt_cp_dtect} under different $\eta$ values and different online identifiers.}
\begin{figure}
    \begin{subfigure}[t]{0.48\columnwidth}
        \centering
        \includegraphics[width=\columnwidth]{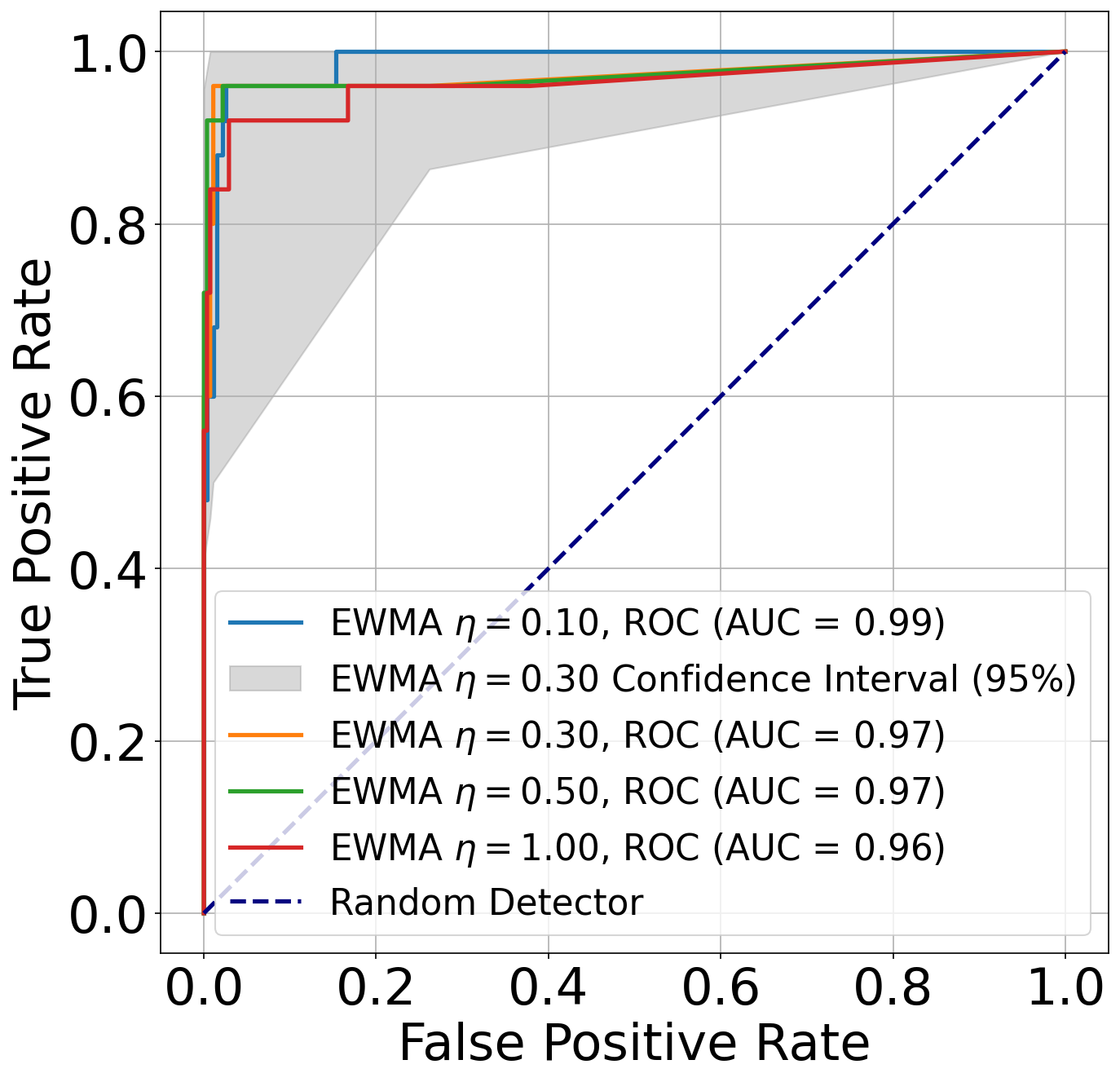}
        \caption{EF-RLS change point detection ROC curve}
        \label{fig:rls_roc}
        \vspace*{2mm}
    \end{subfigure}
    \begin{subfigure}[t]{0.48\columnwidth}
        \centering
        \includegraphics[width=\columnwidth]{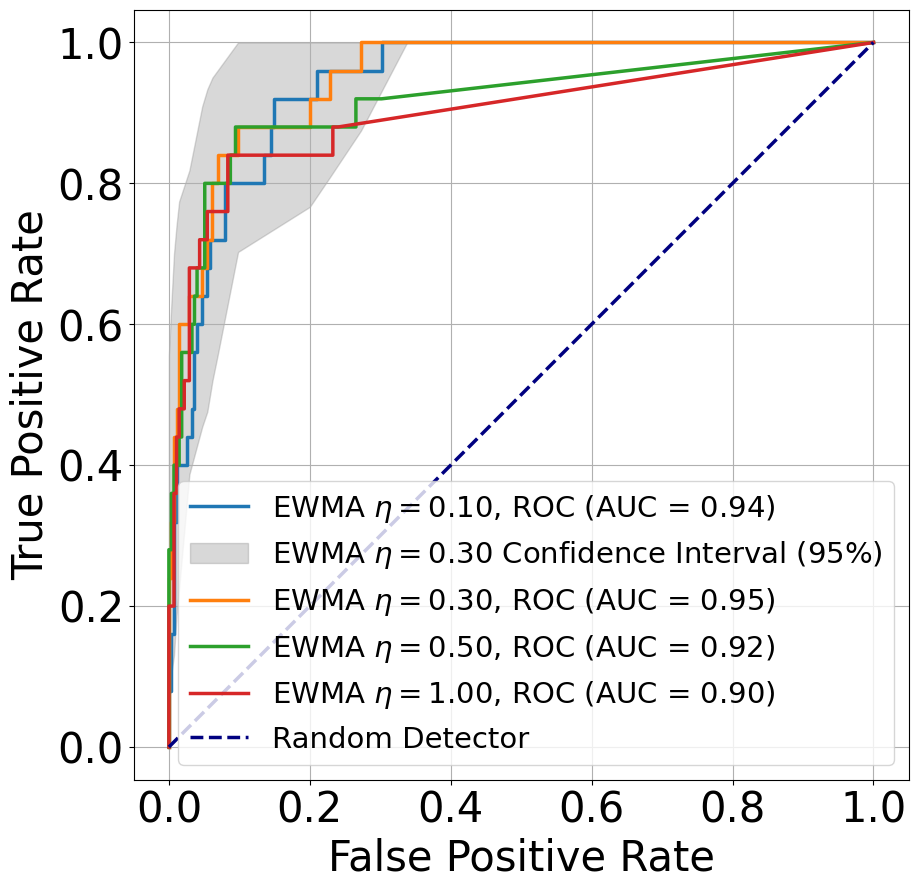}
        \caption{\edit{GW-RLS} change point detection ROC curve}
        \label{fig:gwrls_roc}
        \vspace*{2mm}
    \end{subfigure}
    \caption{ROC curves that capture the performance of change-point detection for the EF-RLS online estimator~(\subref{fig:rls_roc}) and \edit{GW-RLS} online estimator~(\subref{fig:gwrls_roc}).
    }
    \label{fig:ROC}
\end{figure}
Fig.~\ref{fig:ROC} shows the ROC curve of the performance of the change point detection algorithm on EF-RLS and \edit{GW-RLS}.
\edit{When the frequency of changes is high, we observe that EF-RLS performs better when coupled with Algorithm~\ref{alg:rec_lrt_cp_dtect}.}
\edit{Furthermore}, the ROC curve indicates that we can set $\eta=0.3$ to maximize the area under the ROC curve for \edit{GW-RLS}.

\subsection{Time-Varying System}\label{sec:sim:timevarying}
In this section, we examine the effectiveness of the change point detection algorithm proposed in Section~\ref{sec:fault-dect} by assessing parameter estimation errors. 
We start by considering 
Example~\ref{ex:sir_cp} with process noise that is similar to the experimental setup in the previous section, \edit{with step size $h=0.2$}.
We also added observation noise to the observed signal in the form of \eqref{eqn:gen_nonlinear_sys_pre} with a strength of around one-tenth the value of $\psi_k$.
Fig.~\ref{fig:2-node-time-v-inf-data} 
\edit{shows the infection prevalence of the two-node SIR epidemic model with process noise, as specified in Example~\ref{ex:sir_cp} with noise structure similar to \eqref{label:sir_network_sde}.}

\begin{figure}
    \centering
    \includegraphics[width=\columnwidth]{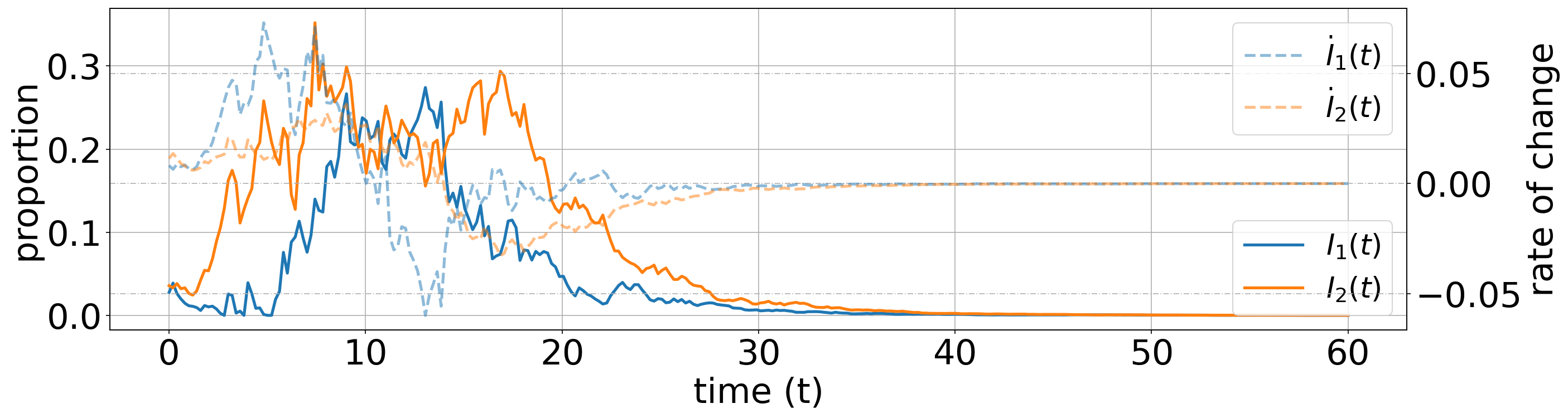}
    \caption{Infection prevalence of a 2-node SIR network}
    \label{fig:2-node-time-v-inf-data}
\end{figure}
\begin{figure}
    \centering
    \begin{subfigure}[t]{\columnwidth}
        \centering
        \includegraphics[width=\columnwidth]{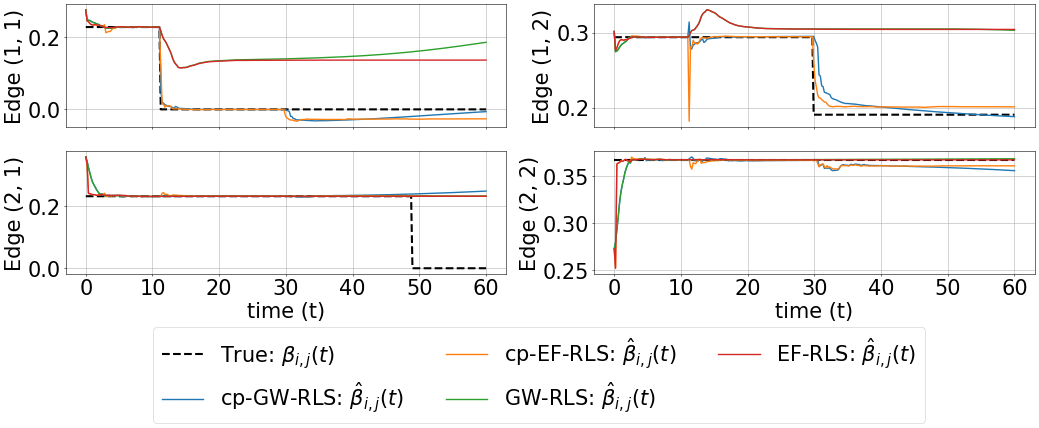}
        \caption{Infection rate estimation}
        \label{fig:time_varying_tracking_performance_beta}
    \end{subfigure}
    \begin{subfigure}[t]{\columnwidth}
        \centering
    \end{subfigure}
    \begin{subfigure}[t]{\columnwidth}
        \centering
        \includegraphics[width=\columnwidth]{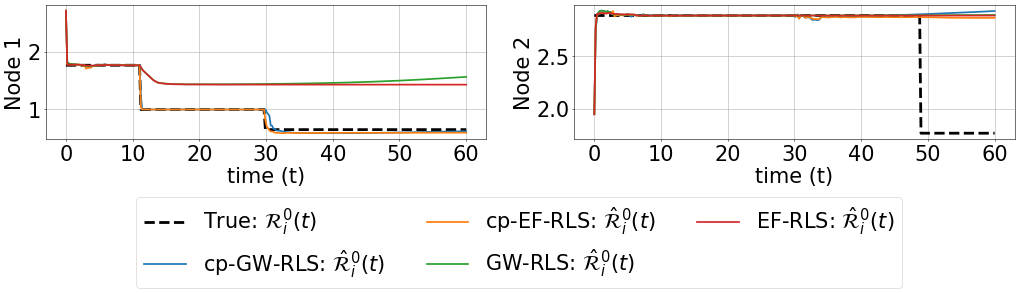}
        \caption{Local basic reproduction number estimation}
        \label{fig:time_varying_tracking_performance_r0}
    \end{subfigure}
    \caption{Parameter tracking performance comparison before and after applying change point detection resetting.
    \edit{The cp-GW-RLS and cp-EF-RLS represent the estimation trajectory under GW-RLS and EF-RLS when coupled with the proposed change point state resetting method described in Algorithm~\ref{alg:rec_lrt_cp_dtect}.}
    }
    \label{fig:time_varying_tracking_performance}
\end{figure}

\edit{
Fig.~\ref{fig:time_varying_tracking_performance} demonstrates that the parameter estimation performance significantly improves when coupled with the change point state resetting scheme. In Fig.~\ref{fig:time_varying_tracking_performance_beta}, we observe that while EF-RLS and GW-RLS respond to the first parameter jump, their adaptive behavior is further enhanced when combined with the proposed change point state resetting algorithm. 
Coupling with Algorithm~\ref{alg:cps_reset_scheme} allows cp-EF-RLS and cp-GW-RLS to successfully adapt to both the first and second parameter changes. 
Notably, the second change point occurs at $t=30$, coinciding with the state signal decay, whereas the final change point at $t=49$ is nearly undetectable.
}
\begin{figure}
    \centering
    \includegraphics[width=\columnwidth]{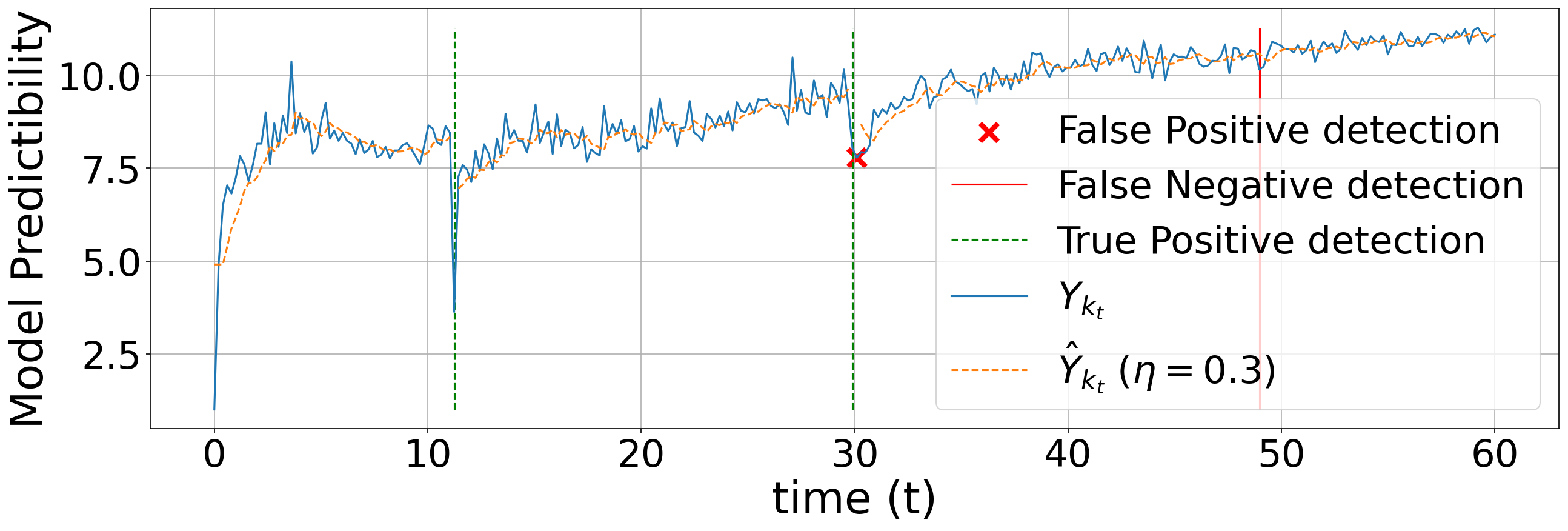}
    \caption{Change point detection performance with GW-RLS, significance level is chosen to be $\tau= 0.1$ with exponential weight $\eta=0.5$.
    The y-axis represents the magnitude of the model predictability $Y_k$, with higher values being preferable.
    }
    \label{fig:chp:rec_lrt_ch_performance}
\end{figure}

Fig.~\ref{fig:chp:rec_lrt_ch_performance} shows the performance of Algorithm~\ref{alg:rec_lrt_cp_dtect} in terms of the true-positive, false-positive, and false-negative detection \edit{instances}. \edit{When the frequency of the parameter changes is low, Algorithm~\ref{alg:rec_lrt_cp_dtect} performs relatively well, and the effect of false positive and negative detections does not negatively impact future predictability, in this example.}

\edit{
Next, we are interested in the change point detection performance under different hyperparameter setups.
In particular, we illustrate the intuition behind choosing different $\eta$ and its impact on the change point detection performance.
}
\begin{figure}
    \centering
    \includegraphics[width=\columnwidth]{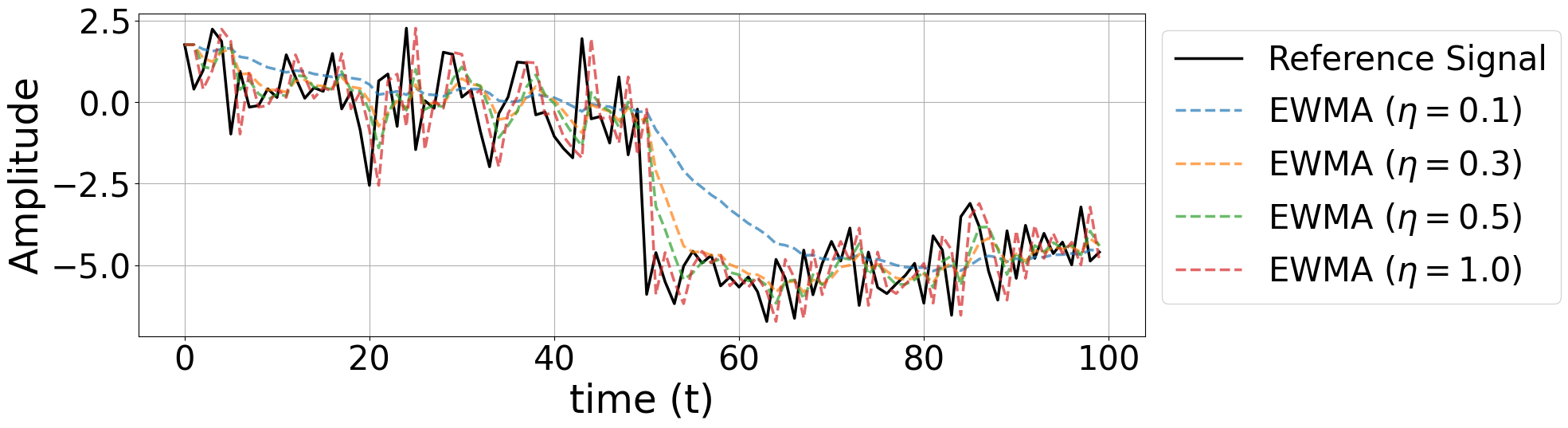}
    \caption{The EWMA acts as a low pass filter on a Gaussian white noise signal with an abrupt downward drift at time $t=50$.
    Values of $\eta \in \{0.1, 0.3, 0.5, 1.0\}$ determine the smoothing effect of the EWMA.
    }
    \label{fig:chpt:EWAM_lagged_effects}
\end{figure}
In Fig.~\ref{fig:chpt:EWAM_lagged_effects}, we \edit{observe} that the exponential weight $\eta \in [0, 1]$ determines the smoothing intensity of an EWMA.
When $\eta = 1$, EWMA corresponds to the most recent update and provides no smoothing effect.
As $\eta$ approaches $0$, EWMA relies solely on its initial value, resulting in maximal smoothing. 

\edit{Next, we increase the number of change points and check the performance of Algorithm~\ref{alg:rec_lrt_cp_dtect} under different $\eta$ values and different online identifiers.}
\begin{figure}
    \begin{subfigure}[t]{0.48\columnwidth}
        \centering
        \includegraphics[width=\columnwidth]{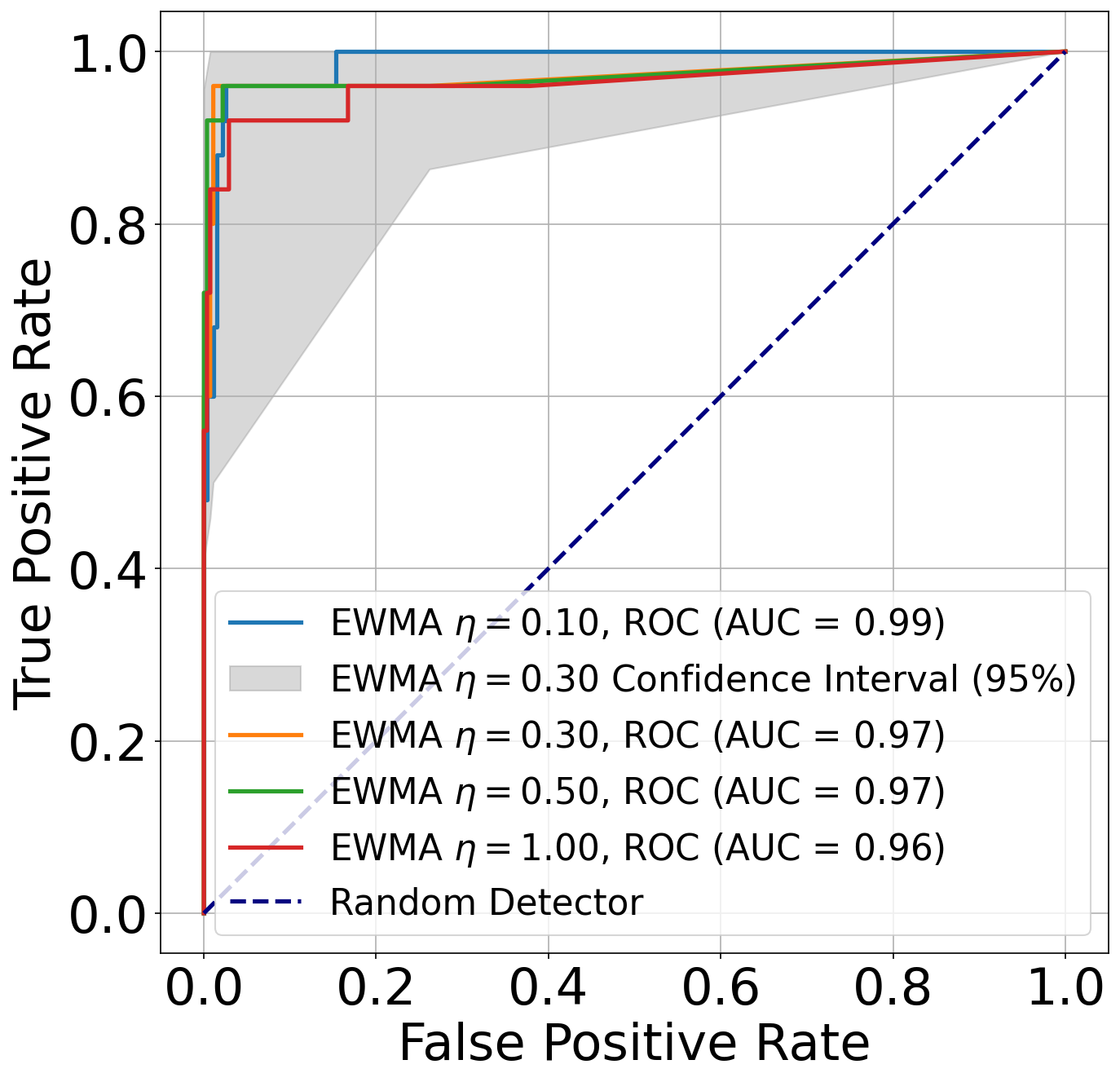}
        \caption{EF-RLS change point detection ROC curve}
        \label{fig:rls_roc}
        \vspace*{2mm}
    \end{subfigure}
    \begin{subfigure}[t]{0.48\columnwidth}
        \centering
        \includegraphics[width=\columnwidth]{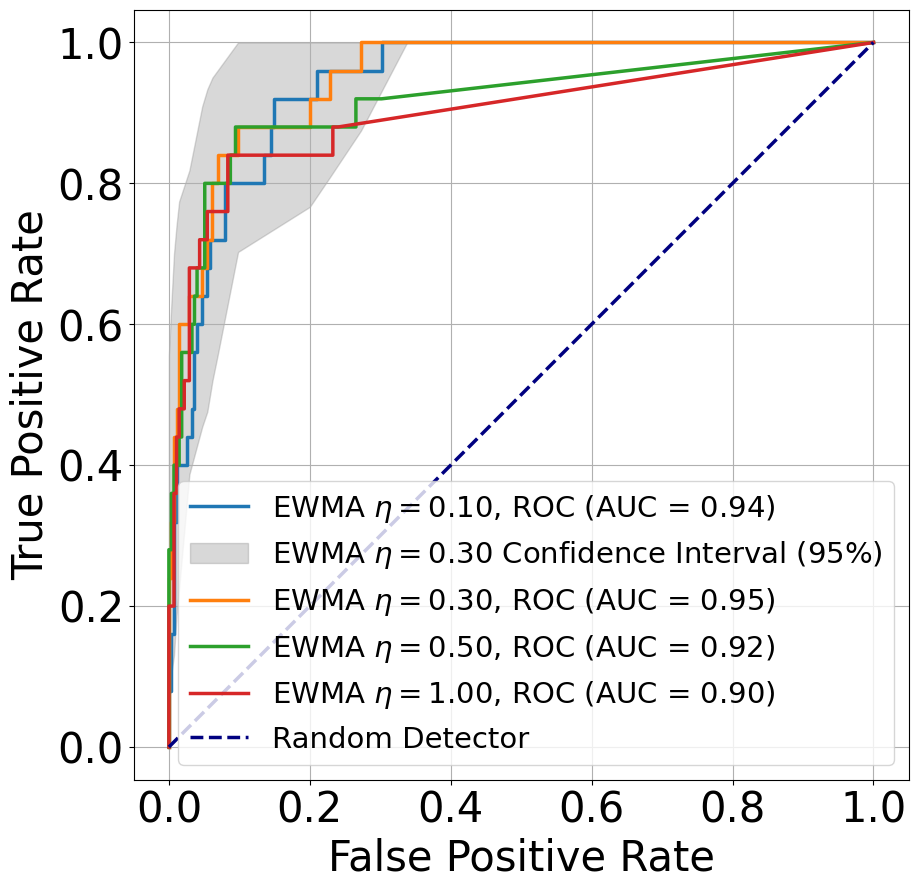}
        \caption{\edit{GW-RLS} change point detection ROC curve}
        \label{fig:gwrls_roc}
        \vspace*{2mm}
    \end{subfigure}
    \caption{ROC curves that capture the performance of change-point detection for the EF-RLS online estimator~(\subref{fig:rls_roc}) and \edit{GW-RLS} online estimator~(\subref{fig:gwrls_roc}).
    }
    \label{fig:ROC}
\end{figure}
Fig.~\ref{fig:ROC} shows the ROC curve of the performance of the change point detection algorithm on EF-RLS and \edit{GW-RLS}.
\edit{When the frequency of changes is high, we observe that EF-RLS performs better when coupled with Algorithm~\ref{alg:rec_lrt_cp_dtect}.}
\edit{Furthermore}, the ROC curve indicates that we can set $\eta=0.3$ to maximize the area under the ROC curve for \edit{GW-RLS}.

\section{Conclusion}
\edit{In this work, we have identified two significant challenges that hinder the effectiveness of online parameter estimation methods in nonlinear systems: the lack of persistent excitation and  practical non-identifiability, and illustrate them with  epidemic models as examples. 
To address these issues, we introduced the Greedily-Weighted Recursive Least Squares (GW-RLS) algorithm, which leverages recursive least squares combined with the concept of an excitation set to ensure that the data points used for regression provide sufficient informational richness.}

\edit{
The proposed GW-RLS algorithm demonstrates outstanding performance over conventional methods, successfully identifying epidemic parameters within networked SIS and SIR models, even in the presence of process and observation noise. 
Additionally, we introduced a change-point detection-based memory resetting scheme, enhancing the algorithm's capability to accurately track parameters that vary over time via abrupt changes. 
We also demonstrate that the change-point detection-based memory resetting scheme works well with other online estimation methods, including EF-RLS. 
The effectiveness of this parameter-tracking enhancement is validated through extensive numerical simulations.}

\edit{
Despite its effectiveness, the method has limitations: selecting points by minimizing the information-matrix condition number can be noise-sensitive (regularization or constrained updates may help); it favors abrupt parameter changes and may miss slow dynamics (adaptive or hybrid strategies could address this); and it still requires real-world validation. 
Future work will add the least singular value alongside the condition number in the greedy selection, improve robustness, and evaluate performance on real datasets.
}

\bibliographystyle{IEEEtran}
\bibliography{./bib/refs}

\end{document}